\documentclass[aps,pra,amsmath,amssymb,twocolumn,nofootinbib,floatfix,superscriptaddress,10pt,showpacs]{revtex4-1}
\usepackage{amsthm}
\usepackage{latexsym}
\usepackage{amsfonts}
\usepackage{bbm,dsfont}
\usepackage{graphicx}
\usepackage[usenames,dvipsnames]{xcolor}

\def\???{{\color{red}???}}
\definecolor{orange}{rgb}{1.0,0.3,0}

\def\<{\langle}
\def\>{\rangle}

\newcommand{\be}{\begin{eqnarray} \begin{aligned}}
\newcommand{\ee}{\end{aligned} \end{eqnarray} }
\newcommand{\benn}{\begin{eqnarray*} \begin{aligned}}
\newcommand{\eenn}{\end{aligned} \end{eqnarray*} }

\newcommand{\ben}{\begin{eqnarray} \begin{aligned}}
\newcommand{\een}{\end{aligned} \end{eqnarray} }

\newcommand{\bc}{\begin{center}}
\newcommand{\ec}{\end{center}}

\newcommand{\norm}[1]{\left\| #1\right \|}

\newcommand{\beq}{\begin{eqnarray} \begin{aligned}}
\newcommand{\eeq}{\end{aligned} \end{eqnarray} }
\newcommand{\bea}{\begin{array}}
\newcommand{\eea}{\end{array}}

\newcommand{\bee}{\begin{enumerate}}
\newcommand{\eee}{\end{enumerate}}
\newcommand{\bei}{\begin{itemize}}
\newcommand{\eei}{\end{itemize}}
\newtheorem{theorem}{Theorem}
\newtheorem{proposition}[theorem]{Proposition}

\newtheorem{lemma}[theorem]{Lemma}

\newtheorem{definition}[theorem]{Definition}

\theoremstyle{remark}
\newtheorem{remark}[theorem]{Remark}

\newcommand{\hil}{\mathcal{H}}

\usepackage{amsfonts}

\def\01{\{0,1\}}

\newcommand{\eps}{\varepsilon}
\newcommand{\ket}[1]{|#1\rangle}
\newcommand{\bra}[1]{\langle#1|}

\def\<{\langle}
\def\>{\rangle}

\newcommand{\newreptheorem}[2]{%
\newenvironment{rep#1}[1]{%
 \def\rep@title{#2 \ref{##1} (restatement)}%
 \begin{rep@theorem}}%
 {\end{rep@theorem}}}
\makeatother

\newreptheorem{thm}{Theorem}
\newreptheorem{lem}{Lemma}

\def\paragraph#1{%
  \smallskip%
  \par\noindent{\textbf{#1}}\quad
}

\let\tr\Tr

\newcommand{\R}{\mathbb{R}}
\newcommand{\C}{\mathbb{C}}

\newcommand{\tr}[1]{\mathrm{Tr}\left[ {#1} \right]} 
\newcommand{\Tr}[2]{\mathrm{Tr}_{#1}\left[ {#2} \right]} 
\newcommand{\asmpequiv}{\asymp}

\newcommand{\N}{\mathbb{N}}
\newcommand{\Q}{\mathbb{Q}}
\newcommand{\beqn}{\begin{equation}}
\newcommand{\eeqn}{\end{equation}}
\definecolor{ultramarine}{RGB}{63, 0, 255}
\definecolor{medblue}{RGB}{0, 0, 100}
\definecolor{panblue}{RGB}{0,24,150}
\definecolor{carmine}{RGB}{150, 0, 24}
\usepackage{hyperref}
\hypersetup{colorlinks,
linkcolor=carmine,
citecolor=medblue,
urlcolor=panblue,
anchorcolor=OliveGreen}
\usepackage{bbm} 

\newcommand{\anc}{\mathrm{anc}}
\newcommand{\poly}{\mathrm{poly}}
\newcommand{\spec}{\mathrm{sp}}
\let\originalleft\left
\let\originalright\right
\renewcommand{\left}{\mathopen{}\mathclose\bgroup\originalleft}
\renewcommand{\right}{\aftergroup\egroup\originalright}

\begin{document}
\title{A Resource Theory for Work and Heat}
\author{Carlo Sparaciari}\email{carlo.sparaciari.14@ucl.ac.uk}
\affiliation{Department of Physics and Astronomy, University College London,
London WC1E 6BT, United Kingdom}
\author{Jonathan Oppenheim}\email{j.oppenheim@ucl.ac.uk}
\affiliation{Department of Physics and Astronomy, University College London,
London WC1E 6BT, United Kingdom}
\author{Tobias Fritz}\email{fritz@mis.mpg.de}
\affiliation{Max Planck Institute for Mathematics in the Sciences, Leipzig, Germany, and Perimeter Institute for Theoretical Physics, Waterloo, Canada}
\begin{abstract}
Several recent results on thermodynamics have been obtained
using the tools of quantum information theory and resource theories.
So far, the resource theories utilised to describe thermodynamics have
assumed the existence of an infinite thermal reservoir, by declaring that thermal states
at some background temperature come for free. Here, we propose a resource theory of
quantum thermodynamics without a background temperature, so that no states at all come
for free. We apply this resource theory to the case of many non-interacting systems,
and show that all quantum states are classified by their entropy and average energy, even arbitrarily far away from equilibrium. This implies that
thermodynamics takes place in a two-dimensional convex set that we call the \emph{energy-entropy diagram}.
The answers to many resource-theoretic questions about thermodynamics can be read off from
this diagram, such as the efficiency of a heat engine consisting of finite reservoirs, or the rate of conversion between two states.
This allows us to consider a resource theory which puts work and heat on an equal footing, and serves as a model for other resource theories.
\end{abstract}
\date{\today}
\pacs{05.30.-d, 05.70.-a, 03.67.-a}
\maketitle
\section{Introduction}
\label{intro}
To make precise statements about thermodynamics, particular at the quantum scale, we
need to precisely define what thermodynamics is. In particular, we have to specify what an
experimenter is allowed to do to take a system from one state into another. This specification
defines a resource theory, something which has been used to successfully describe
thermodynamic phenomena occurring at the microscopic scale~\cite{janzing_thermodynamic_2000,
horodecki_fundamental_2013,
horodecki_reversible_2003,
brandao2013second,
brandao_resource_2013,
yunger_halpern_beyond_2016,
woods_maximum_2015,
cwiklinski_limitations_2015,
lostaglio_quantum_2015,
wilming_second_2016,
gour_resource_2015,
masanes_derivation_2014,
alhambra_second_2016,
perry_sufficient_2015,
halpern_microcanonical_2015,
alicki_entanglement_2013}. Any line of research which specifies what the rules of thermodynamics are,
can be said to define a resource theory~\cite{crooks1999entropy,
dahlsten_inadequacy_2011,
rio_thermodynamic_2011,
aaberg-singleshot,
faist_minimal_2015,
skrzypczyk_work_2014,
korzekwa_extraction_2016,
narasimhachar_low-temperature_2015,
gallego_thermodynamic_2015,
schulman_molecular_1999}.
These theories typically consist of a state space and a set of allowed operations that can be performed
on the states (see e.g.~\cite{horodecki_quantumness_2012,
coecke_resources_2014,
delrio_knowledge_2015,
vinjanampathy_quantum_2015,
goold_review_2016}
for reviews).
\par
The resource theories developed so far for quantum thermodynamics are based on
assuming that thermal states (Gibbs states) at a fixed background temperature come
for free. In these theories, states are classified by their free energy, and this quantity
also equals the amount of work that can be extracted from many copies of a given
state~\cite{brandao_resource_2013,
aaberg-singleshot,
horodecki_fundamental_2013,
alicki_entanglement_2013,
skrzypczyk_work_2014}.
However, declaring those thermal states to be free boils down to assuming the existence of an infinite thermal reservoir. This cannot always be taken for
granted~\cite{tajima2014optimal,richens_finite_2016}. In some applications, such
as many types of engines, the system under consideration operates on such short
timescales that it must be considered a closed system. In other applications, the
environment is finite and its state changes due to the interaction with the system,
for example when a power plant dumps large amounts of heat to the environment.
It seems therefore imperative to develop thermodynamics as a resource theory
\emph{without} assuming the existence of an infinite thermal reservoir.
This is the aim of the present paper.
\par
The set of allowed operations in our resource theory is much broader
than the one a typical experimentalist can implement, as any energy-preserving unitary
is allowed. Therefore, our theory primarily delineates fundamental limitations to what
is possible in ``real life''. However, we suspect that our abstract achievability results can actually be implemented using a more realistic 
set of operations, as in the resource theory of Thermal Operations~\cite{perry_sufficient_2015},
where only changing the energy levels of the system, and thermal contact with a heat bath is allowed.
But for the time being, our results should be seen as upper limits, only achievable in idealised
conditions.
Moreover, the results that we present within our theory are concerned with the
asymptotic regime, i.e.~the limit of many non-interacting identical system. This follows the
abstract approach to resource theories developed in~\cite{fritz_resource_2015}. Although this is a limitation, we believe that one needs to understand the asymptotic structure of a resource theory first before analyzing the single-shot regime, and this is what we achieve here for thermodynamics.
\par
In addition to providing a general framework for thermodynamics in the absence of an infinite
bath, we prove in Thm.~\ref{TheoSEmain} that two quantum states are asymptotically equivalent under energy-preserving unitaries
if and only if they have same entropy and average energy.
Due to this equivalence, we interpret the specification of entropy and average energy of a
state as the description of a thermodynamic \emph{macrostate}. Thermodynamics in the asymptotic limit
is then studied by only considering the set of macrostates, which we call the \emph{energy-entropy diagram}.
This diagram is a complete description of thermodynamics in the asymptotic limit. All of this takes place arbitrarily far away from equilibrium.

We then use our methods to study rates of conversion between two states of a closed system, and to propose a definition of
the work and heat exchanged while interconverting two states using a finite thermal reservoir and
a battery. The resulting expressions for work and heat, Eqs.~\eqref{final_work_def} and
\eqref{final_heat_def}, recover the standard ones in the limit of an infinite thermal reservoir. 
\section{Framework and allowed operations}
\label{s1:framework}
The systems we consider in our resource theory consist of $n$ copies of a
single $d$-level system described by $\mathbb{C}^d$ (a {\it qudit}) with fixed Hamiltonian $H$;
both $d$ and $H$ are parameters of the theory while $n$ varies. We assume the total Hamiltonian $H_{\text{tot}}$
of the total $n$-copy system to be the sum of single-qudit Hamiltonians $H$, each of them acting on a different copy,
which makes the different copies non-interacting. The resource objects of our theory
are quantum states on such an $n$-qudit system, for arbitrary $n\in\N$. The allowed operations for turning one
such state into another are all the global unitaries $U$ acting on the total system which are energy-preserving, $[U,H_{\text{tot}}]=0$. Thus, we assume perfect control over our closed system,
and the sole limitation is set by the first law of thermodynamics, requiring conservation of energy.
The class of operations is purposely broad, as we are interested in fundamental limitations imposed by the laws
of nature, as opposed to the ones imposed by our limited control over macroscopic systems.
Due to the unitary nature of the operations, the state transformations we consider are reversible by definition.
For simplicity, we assume the Hamiltonian to be fixed throughout, without any possibility of changing it.
In Thm.~\ref{TheoSEmain} we also permit the use of an ancilla system of sublinear size and energy
spectrum which can be initialized in an arbitrary state and gets discarded at the end. While the
ancilla allows us to act more freely over the main system, it does not modify the physical assumptions
made so far. Indeed, since we work in the thermodynamic limit, the sublinearity of the ancilla makes
its energetic and entropic contribution (per single copy of the system) vanishingly small.
Finally, when talking about rates of conversion, we permit discarding subsystems that are decoupled
from the rest. 
\par
Our resource theory can describe both the thermodynamics of closed systems (where the thermal
reservoir is absent), and open systems interacting with a finite thermal reservoir (where the size of the
system is comparable to the one of the reservoir). Thus, our framework extends the one of Thermal
Operations~\cite{janzing_thermodynamic_2000, horodecki_reversible_2003, brandao2013second,
horodecki_fundamental_2013, brandao_resource_2013, aaberg-singleshot}, in which one can
add an arbitrary number of thermal states $\tau_{\beta} = Z_\beta^{-1} e^{-\beta H}$ at a given
temperature $\beta^{-1}$, with $Z_\beta = \tr{e^{-\beta H}}$ the partition function of the system.
Indeed, adding arbitrary many thermal states is equivalent to the system being in
contact with a reservoir with infinite heat capacity. By not allowing this possibility, we obtain a theory
which can describe, among other things, systems in contact with a finite reservoir in which thermal
states are themselves a valuable resource, and recover the case of an infinite reservoir in the limit.
\section{Asymptotic equivalence of quantum states}
\label{s2:adia_equ}
Our resource theory clearly has many conserved quantities: since the allowed operations are all unitaries,
a can be converted into another only if they live on the same number of qudits and have the same
spectrum. Moreover, since our unitaries are energy-preserving, the states must have the same
\emph{distribution} over the energy levels, or equivalently the same moments of energy. This makes our
theory very restrictive.
\par
However, at the asymptotic level the situation is quite different: as it turns out, two quantum states
can be interconverted if and only if they have the same entropy and average energy\footnote{Notice
that these quantities are precisely the \emph{asymptotically continuous}
ones~\cite{synak-radtke_asymptotic_2006}.}. In particular, we say that a state $\rho$ is \emph{asymptotically
equivalent} to another state $\sigma$, and write $\rho\asmpequiv\sigma$ if the equivalent conditions of the following theorem hold.
\begin{theorem}
	For states $\rho$ and $\sigma$ on any quantum system of dimension $d$ with given Hamiltonian $H$, the following are equivalent:
	\begin{enumerate}
		\item\label{equalee} The states have equal entropy and average energy,
			\beqn
			S(\rho) = S(\sigma) \ , \ E(\rho) = E(\sigma).
			\eeqn	
		\item\label{conversion-best} There exists an ancillary system $A$ of $O(\sqrt{n\log n})$ many qudits whose Hamiltonian $H_{A}$ satisfies $\norm{H_{A}} \leq O(n^{2/3})$ with state $\eta$ as well as an energy-preserving unitary $U$ such that
			\beqn
			\label{asympt_eq}
				\norm{\Tr{\mathrm{A}}{U (\rho^{\otimes n}\otimes\eta) U^\dag} - \sigma^{\otimes n}}_1 \stackrel{n\to\infty}{\longrightarrow} 0.
			\eeqn
	\end{enumerate}
	\label{TheoSEmain}
\end{theorem}
\noindent
where $\|X\|_1=\tr{\sqrt{X^{\dagger}X}}$ is the trace norm, $E(\rho)=\tr{H \rho}$ is the average
energy, and $S(\rho)=- \tr{\rho \log \rho}$ is the von Neumann entropy, which was shown to be
equivalent to a thermodynamic entropy~\cite{mathematical_von_Neumann}, and coincides, in
the thermodynamic limit, with the Boltzmann entropy (see for example Ref.~\cite{brandao_resource_2013}
in a similar context as the present work). We prove this theorem in Appendix~\ref{asy_eq_theorems}, and also show that the bounds on the size and Hamiltonian of the ancilla can equivalently be taken to be $o(n)$. This sublinearity is essential: the amount of entropy and energy that can be exchanged between the system and the ancilla tends to $0$ as $n\rightarrow\infty$, when measured per copy of the system.
This sketches the reason why asymptotic equivalence implies equal entropy and
average energy.
\par
To show the opposite direction, we need to specify a protocol which achieves~\eqref{asympt_eq}, turning $\rho^{\otimes n}$ into something close (in trace norm) to $\sigma^{\otimes n}$.
Concretely, our ancilla is composed by three subsystems, each of them
playing a different role in our protocol. The first subsystem provides
a source of randomness, used to modify the spectrum of the state $\rho^{\otimes n}$;
its state is maximally mixed, and its Hamiltonian is trivial. The second subsystem
is used as a register, and allows us to dilate slightly irreversible operations on the
global system to reversible ones. Its initial state is pure, and the Hamiltonian is again trivial. The third subsystem makes the
transformation energy-preserving, and allows for introducing and removing coherence
in the energy eigenbasis. It has a non-trivial Hamiltonian, and its state is in a uniform
superposition of the energy eigenstates. Overall, the ancilla satisfies the properties
listed in \ref{conversion-best}. The full details are given in Appendix~\ref{asy_eq_theorems}.
\par
Thus we can interconvert asymptotically, using the set of allowed operations, between
states with the same entropy $S$ and average energy $E$, in a reversible manner. Consequently, we can classify
any quantum state asymptotically in terms of these two quantities only. Such a passage
from quantum to macroscopic states is at the core of thermodynamics in the guise of the
passage from \emph{microstates} to \emph{macrostates}. Our result seems to capture
this, despite being built on the idealised assumption of non-interacting copies. From now
on, we identify the many-copy limit that one takes when considering asymptotic equivalence
with the standard macroscopic limit of thermodynamics.

\section{The energy-entropy diagram}
\label{en_ent_diag}

Theorem~\ref{TheoSEmain} shows that it is only the energy $E(\rho)$ and entropy $S(\rho)$ of a state $\rho$ that determines its behaviour under many-copies transformations. Hence as far as the many-copies level is concerned, we can \emph{identify} a state $\rho$ with the pair of values $(E(\rho),S(\rho))\in\R^2$. In order to understand thermodynamics as a resource theory asymptotically, we therefore need to ask: which pairs of numbers $x=(x_E,x_S)\in\R^2$ do arise from a state in this manner? We call this set the \emph{energy-entropy diagram}. The energy-entropy diagram depends on the system Hamiltonian $H$, and can be characterised as follows:

\begin{figure}[!ht]
	\centering
	\includegraphics[width=0.95\columnwidth]{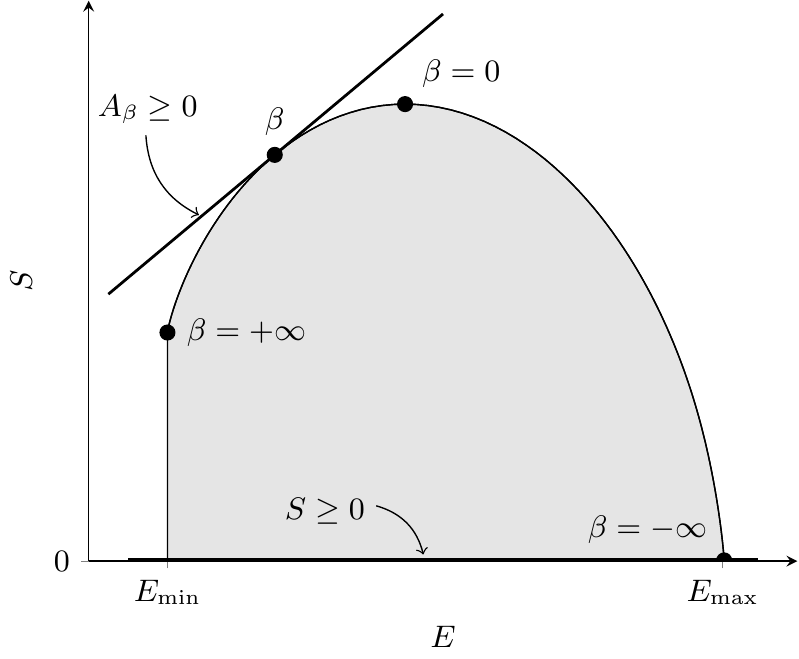}
	\caption{The energy-entropy diagram representing the state space of a quantum
system with Hamiltonian $H$, degenerate in the ground state (vertical line). The
physical points are inside the grey area. Each point $\left( E , S \right)$ represents
an equivalence class of microstates, i.e.~a single macrostate. Ineq.~\eqref{entropy_ineq},
postulating the nonnegativity of entropy, is satisfied by the points above the $E$-axis. For a
given $\beta$, the Ineq.~\eqref{athermality_ineq} is satisfied by those points below the drawn line which is tangent
to the physical region, and goes through the point $\left(E(\tau_{\beta}) , S(\tau_{\beta}) \right)$ with slope $\beta$.}
	\label{fig_en_ent_diag}
\end{figure}

\begin{proposition}
The values $(E(\rho),S(\rho))$ form a closed convex subset of $\R^2$ as in Figure~\ref{fig_en_ent_diag}, where the lower boundary is the line $x_S=0$ and the upper boundary is the curve $\beta\mapsto (E(\tau_\beta),S(\tau_\beta))$ traced by the thermal states $\tau_\beta = Z_\beta^{-1} e^{-\beta H}$ for $\beta\in[-\infty,+\infty]$.
\label{ESconv}
\end{proposition}

Here, we use Cartesian coordinates $x=(x_E,x_S)$ when speaking about a point in $\R^2$ that may or may not belong to the energy-entropy diagram, and we also write $x(\rho):=(E(\rho),S(\rho))$ for the point associated to a specific state $\rho$. Thus we have the components $x(\rho)_E = E(\rho)$ and $x(\rho)_S = S(\rho)$, and the curve of thermal states is given by $\beta\mapsto x(\tau_\beta)$.

\begin{proof}
	It is a standard fact that the states of maximal entropy for a given energy are precisely the thermal states. One way to see that this also holds for $\beta < 0$ is by using the fact that it holds for $\beta > 0$ and reversing the sign of the Hamiltonian.

	It remains to show that also every smaller value of the entropy is achievable for a given energy. Clearly $S=0$ is achievable, namely by considering a pure state in a suitable superposition of energy levels which has the desired expectation value of energy. Moreover, the set of all states of a given energy is a convex subset of all the density matrices, and in particular it is path-connected; since the map $\rho\mapsto S(\rho)$ is continuous, it follows that also its image under $S$ is path-connected, so that also all intermediate entropy values can be achieved for the given energy.

	Convexity follows from the alternative characterisation in terms of linear inequalities that we will derive as Proposition~\ref{ESfacets}.
\end{proof}

\begin{remark}
\label{sloperemark}
The curve of thermal states is parametrized by $\beta\mapsto (E(\tau_\beta),S(\tau_\beta))$, where
\begin{align}
	E(\tau_\beta) &= \tr{H \frac{e^{-\beta H}}{Z_\beta}} = - \frac{d\log Z_\beta}{d\beta}, \\
	S(\tau_\beta) &= - \tr{\frac{e^{-\beta H}}{Z_\beta} \log ( \frac{e^{-\beta H}}{Z_\beta} ) } =
	\beta E(\tau_\beta) + \log Z_\beta
\end{align}
Differentiating with respect to $\beta$ and collecting terms results in the fundamental thermodynamic relation $dS(\tau_\beta) = \beta\, dE(\tau_\beta)$. For the energy-entropy diagram, this implies that the parameter $\beta$ is precisely equal to the slope of the tangent at each point on the curve of thermal states.
\end{remark}

By virtue of being convex and topologically closed, one can describe the energy-entropy diagram also in a dual way by writing down all the linear inequalities that bound it. These inequalities are most conveniently stated in terms of the quantities
\beqn
	\label{athermality}
	A_\beta(x) := \beta x_E - x_S + \log Z_\beta,
\eeqn
so that $A_\beta(\tau_\beta) = 0$. Note that $A_\beta(x)$ is a linear function of $x$. 

For a given value of $\beta$, we call $A_\beta$ the \emph{$\beta$-athermality}\footnote{This terminology was suggested to us by Matteo Smerlak. Without the additive constant, $A_\beta$ has also been called the ``free entropy'' in~\cite{guryanova_thermodynamics_2015}.}, since it vanishes on the thermal state $x(\tau_\beta)$, and we think of $A_\beta(x)$ as a measure of how far $x$ is from being equal to $x(\tau_\beta)$. The $\beta$-athermality differs from the free energy $x_E - \beta^{-1} x_S$ only by an additional factor of $\beta$ and an additive constant. One of the reasons that we prefer using~\eqref{athermality} over the free energy is that on a state $\rho$, we can also neatly write it as the relative entropy distance to the thermal state,
\beqn
	\label{athermrelent}
	A_\beta(x(\rho)) = D(\rho\|\tau_\beta),
\eeqn
as one can see by writing out the definition of relative entropy and plugging in $\tau_\beta = Z_\beta^{-1} e^{-\beta H}$. This again justifies the term ``$\beta$-athermality''. For $\beta=0$, we obtain the negentropy $A_0(x(\rho)) = \log d - S(\rho)$.

We can now state the characterisation of the energy-entropy diagram by linear inequalities:

\begin{proposition}
	The energy-entropy diagram is the set of all points $x = (x_S, x_E) \in \R^2$ such that
	\begin{align}
	\label{entropy_ineq}
	x_S &\geq 0, \\
	\label{athermality_ineq}
	A_\beta(x) &\geq 0 \ \ \ \forall \, \beta \in \R.
	\end{align}
	\label{ESfacets}
\end{proposition}

\begin{proof}
	All the inequalities hold for an achievable point $x(\rho) = (E(\rho),S(\rho))$, since both the entropy and the relative entropy~\eqref{athermrelent} are nonnegative.

	In the other direction, we need to show that if $x\in\R^2$ satisfies all the claimed inequalities, then it lies in the energy-entropy diagram. So by assumption, we have $x_S\geq 0$ and
	\[
		\beta x_E - x_S \geq \beta E(\tau_\beta) - S(\tau_\beta)
	\]
	for all $\beta\in\R$. Taking $\beta\to +\infty$ and $\beta\to-\infty$ shows that we also must have $E_{\min}\leq x_E\leq E_{\max}$. Since there is a unique thermal state at every given energy in this range, there is a unique $\beta$ such that $E(\tau_\beta) = x_E$. Using this $\beta$, we obtain from the previous inequality
	\[
		x_S \leq S(\tau_\beta) + \beta(x_E - E(\tau_\beta)) = S(\tau_\beta),
	\]
	so that $x$ lies indeed below the curve traced by the thermal states in Figure~\ref{fig_en_ent_diag}.
\end{proof}

\begin{remark}
In Propositions~\ref{ESconv} and~\ref{ESfacets}, and in some of our upcoming results, we also consider negative temperatures,
that is, $\beta < 0$. This is due to the fact that our systems are finite; for infinite systems, there is no $E_{\max}$ and no thermal state at $\beta=0$, meaning that the diagram is unbounded to the right and bounded by the thermal curve which increases forever. While we expect our methods to still apply in this case, our present proofs only hold for finite systems.

While thermal states at a negative temperature do not usually arise as a result of thermalisation, they
still play an important roles in multiple physical effects (such as, for instance, in lasers, where coherent light amplification
is obtained through population inversion)~\cite{ramsey_negative_1956}. The main difference between thermal states at
$\beta > 0$ and $\beta < 0$, in our theory, is that the former are completely passive states from which we cannot extract
energy by means of unitary operations, while the latter are active states, from which energy can be extracted.

\end{remark}

\section{Thermodynamic variables and the convex cone of macrostates}
\label{convex_cone}

In this section, we introduce an additional macroscopic quantity to characterise the state of a thermodynamic
system, referred to as \emph{system size} or \emph{amount of substance}. With this parameter we can fully characterise thermodynamic transformations on any number of copies of the system, and later also allow for discarding subsystems, so that the number of systems involved changes. So let's consider what happens at the many-copies level asymptotically.

\begin{proposition}
	\label{ndiagram}
	For any $n\in\N$, the energy-entropy diagram of the $n$-copy Hamiltonian $H^{(n)}$ equals the energy-entropy diagram of $H$, scaled up by a factor of $n$.
\end{proposition}

\begin{proof}
	Since $E(\rho^{\otimes n}) = n E(\rho)$, and similarly for $S$, it is clear that the energy-entropy diagram of $H$, when scaled by $n$, is contained in the energy-entropy diagram of $H^{(n)}$. The converse follows from Proposition~\ref{ESconv}, because any thermal state of $H^{(n)}$ is an $n$-th tensor power of a thermal state of $H$.
\end{proof}

Together with Theorem~\ref{TheoSEmain}, this also implies that for every $n$-system state $\rho$ there is a single-system state
$\sigma$ such that $\rho\asmpequiv\sigma^{\otimes n}$
(that is, $\rho$ is asymptotically equivalent to $\sigma^{\otimes n}$), although $\rho$ itself may be arbitrarily far
from being a product state.

In order to keep track of $n=n(\rho)$, the number of copies of the system on which a state $\rho$ lives, it is useful to consider the \emph{triple} of numbers $(E(\rho),S(\rho),n(\rho)))\in\R^3$, for which we also write $y(\rho)$. Each component of this triple is an additive function of $\rho$, and therefore
\beqn
\label{yadd}
	y(\rho\otimes\sigma) = y(\rho) + y(\sigma).	
\eeqn
By Theorem~\ref{TheoSEmain}, the three components of $y(\rho)$ provide a complete classification of single-system and multi-system states in thermodynamics---with given single-system Hamiltonian $H$---up to asymptotic equivalence.

\newcommand{\Therm}{\mathsf{Therm}}

Now we could consider the set of all points $y=(y_E,y_S,y_n)$ that are of the form $y=y(\rho)$
for some state $\rho$, and call it the \emph{energy-entropy-size diagram} associated to the
Hamiltonian $H$. But all of our results are only up to asymptotic equivalence, so that we
effectively only consider states $\rho$ with $n(\rho) \gg 1$. Equivalently, we can also work with
small values of $n$, but then forget that $n$ is required to be an integer by pretending that the
system size can be an arbitrary nonnegative real number. We then still use the symbol ``$n$'',
although it now plays the role of an ``amount of substance'', just as in the ideal gas law $pV=nRT$.
In principle, converting between number of microsystems and amount of
substance involves rescaling by the Avogadro constant. This is another perspective on what we are doing here,
except that we choose to measure the amount of substance with the unit in which the Avogadro
constant is equal to 1. Based on this intuition, we thus define:

\begin{definition}
	The convex cone $\Therm(H)$ consists of all points $y\in\R^3$ that are of the form $y=n\cdot (x_E,x_S,1)$ for some $n\in\R_{\geq 0}$ and $(x_E,x_S)$ in the energy-entropy diagram.
\end{definition}

In other words, $\Therm(H)$ is the convex cone that we obtain by taking the energy-entropy diagram in $\R^2$ and applying the standard ``homogenisation'' trick for turning a convex set into a convex cone by adding an additional coordinate~\cite[p.~31]{polytopes}. We call a point $y\in\Therm(H)$ \emph{normalised} if $y_n = 1$. Every nonzero $y\in\Therm(H)$ is a unique scalar multiple of a normalised point, so that for most purposes it is sufficient to consider normalised points only (see Section~\ref{decomp}).

\begin{remark}
	\label{slicetherm}
	If we slice $\Therm(H)$ at constant third coordinate $n\in\N$ by considering all points of the form $(x_E,x_S,n)\in\Therm(H)$, then this set is precisely the energy-entropy diagram of $H^{(n)}$ thanks to Proposition~\ref{ndiagram}.
\end{remark}

\begin{remark}
	Taking every (multi-system) state $\rho$ to be represented by a point $y(\rho)\in\Therm(H)$ is a standard construction of thermodynamics: it corresponds to passing from the microstate to the macrostate. The thermodynamic variables of a macrostate are precisely\footnote{Of course this depends on which observables are considered to be conserved quantities. For us, as indicated by Theorem~\ref{SEthm}, energy is assumed to be the only observable that is conserved.} energy $E$, entropy $S$, and system size $n$, and the macrostate is specified completely by these three numbers. If one identifies the passage from microstate to macrostate with the information-theoretic many-copies limit, then our Theorem~\ref{TheoSEmain} offers a mathematically rigorous explanation for \emph{why} the macroscopic variables are exactly these three and no others. There \emph{are} many other extensive quantities that are invariant under energy-preserving unitaries---take the R\'enyi entropies or the variance of energy as examples.
These quantities would indeed be relevant also macroscopically if we had required an \emph{exact} conversion of $\rho^{\otimes n}$ into $\sigma^{\otimes n}$ for some $n$, possibly together with a sublinear ancilla. But our definition of many-copy equivalence allows for approximate conversions that become closer and closer to exact as $n\to\infty$. This is a more permissive notion of asymptotic equivalence, under which correspondingly fewer quantities are invariant, namely only the ones that are asymptotically continuous~\cite{synak-radtke_asymptotic_2006}. In the language of~\cite{coecke_resources_2014,fritz_resource_2015}, allowing such approximate conversions introduces an ``epsilonification''.

One may wonder how it is possible that the passage from microstate to macrostate within our idealised theory yields results that are so close to the standard one.
For example, the class of allowed operations considered in our model is extremely wide, and moreover our results are
valid only in the many-copy limit.
This is at least partly explained by noting that the many-copy limit is a faithful enough description (at the macroscopic level) since the
interactions between particles are small (they have an area scaling), compared to the extensive
quantities (which have a volume scaling), as the number of particles grows to infinity. Moreover, results such
as the von Neumann's quantum ergodic theorem~\cite{goldstein_ergodic_2010} help explain this phenomena.

\end{remark}

Returning to technical developments, we extend the $\beta$-athermalities from $\R^2$ to $\R^3$ by setting
\beqn
\label{Abeta}
	A_\beta(y) := \beta y_E - y_S + y_n \log Z_\beta.
\eeqn
On the energy-entropy diagram, which is embedded in $\Therm(H)$ as the set of all normalised points, this coincides with our previous definition of $A_\beta(x)$, and from there we have extended linearly.

On an actual state $\rho$, we can again express the $\beta$-athermality as a relative entropy distance,
\[
	A_\beta(y(\rho)) = D(\rho\|\tau_\beta^{\otimes n(\rho)}),
\]
where the $n(\rho)$ appears because one needs to consider the thermal state on a suitable number of copies of the system.

The characterisation of the energy-entropy diagram by linear inequalities extends easily to $\Therm(H)$:

\begin{proposition}
	The convex cone $\Therm(H)$ is the set of all $y=(y_E,y_S,y_n)\in\R^3$ such that $y_S\geq 0$ and $A_\beta(y)\geq 0$ for all $\beta\in(-\infty,+\infty)$.
	\label{conechar}
\end{proposition}

\begin{proof}
	Any point in $\Therm(H)$ satisfies these inequalities thanks to Proposition~\ref{ESfacets} together with the fact that $A_\beta(\lambda y) = \lambda A_\beta(y)$ for all $\lambda > 0$, so that it is sufficient to consider normalised points $y\in\Therm(H)$ only. Conversely, suppose that $y\in\R^3$ satisfies all these inequalities. Then from $y_S\geq 0$ and $A_0(y) = y_n \log d - y_S\geq 0$ we conclude $y_n\geq 0$. If it is the case that $y_n = 0$, then we conclude $y_S = 0$, and then also $y_E = 0$ from considering $A_\beta(y)\geq 0$ in the two limits $\beta\to\pm\infty$. Otherwise we have $y_n >  0$, and the point $y_n^{-1} (y_E,y_S)$ satisfies all the inequalities necessary to lie in the energy-entropy diagram by Proposition~\ref{ESfacets}, and therefore $y\in\Therm(H)$.
\end{proof}

What this says is that there are two kinds of additive resource monotones that are relevant to thermodynamics:
\begin{itemize}
	\item The entropy function $\rho\mapsto S(\rho)$;
	\item The $\beta$-athermality functions $\rho\mapsto A_\beta(\rho)$ indexed by $\beta\in(-\infty,+\infty)$.
\end{itemize}

In the terminology of~\cite[Section~7]{fritz_resource_2015}, these are \emph{extremal} monotones. There are two more extremal monotones that one obtains by considering $A_\beta$ as $\beta\to\pm\infty$, which results in the two functions
\[
	\rho\mapsto E(\rho) - n(\rho) E_{\min},\qquad \rho\mapsto n(\rho) E_{\max} - E(\rho).
\]
It follows by~\cite[Corollary~7.9]{fritz_resource_2015} that every other additive (and suitably continuous) monotone is a nonnegative linear combination or integral of these extremal ones.

\section{Macroscopic thermodynamics as a general probabilistic theory}
\label{decomp}

As we will illustrate in the upcoming sections, pretty much any resource-theoretic question about macroscopic thermodynamics can be formulated and answered within the convex cone picture that we have developed. However, since the cone $\Therm(H)\subseteq\R^3$ may be a bit challenging to visualise, it helps the intuition to represent any macrostate $y=(y_S,y_E,y_n)\in\Therm(H)$ by the corresponding normalised macrostate $x:=y_n^{-1}(y_E,y_S)$ in the energy-entropy diagram, equipped with a weight of $y_n$. In this picture, combining systems as in~\eqref{yadd} corresponds to taking a convex combination of normalised macrostates, in the sense that
\begin{equation}
\label{comb_macrostates}
	\frac{y(\rho\otimes\sigma)}{n(\rho\otimes\sigma)} = \frac{n(\rho)}{n(\rho\otimes\sigma)}\cdot\frac{y(\rho)}{n(\rho)} + \frac{n(\sigma)}{n(\rho\otimes\sigma)}\cdot\frac{y(\sigma)}{n(\sigma)},
\end{equation}
where $\tfrac{y(\rho\otimes\sigma)}{n(\rho\otimes\sigma)}$, $\tfrac{y(\rho)}{n(\rho)}$ and $\tfrac{y(\sigma)}{n(\sigma)}$ are all normalised macrostates. Normalising by dividing by system size turns the energy and entropy coordinates, which are extensive quantities, into intensive quantities. While extensive quantities combine across subsystems additively, the associated intensive ones combine across subsystems via convex combinations with coefficients given by the relative subsystem sizes (Figure~\ref{comb_fig}).
\begin{figure}[!ht]
\centering
\includegraphics[width=1\columnwidth]{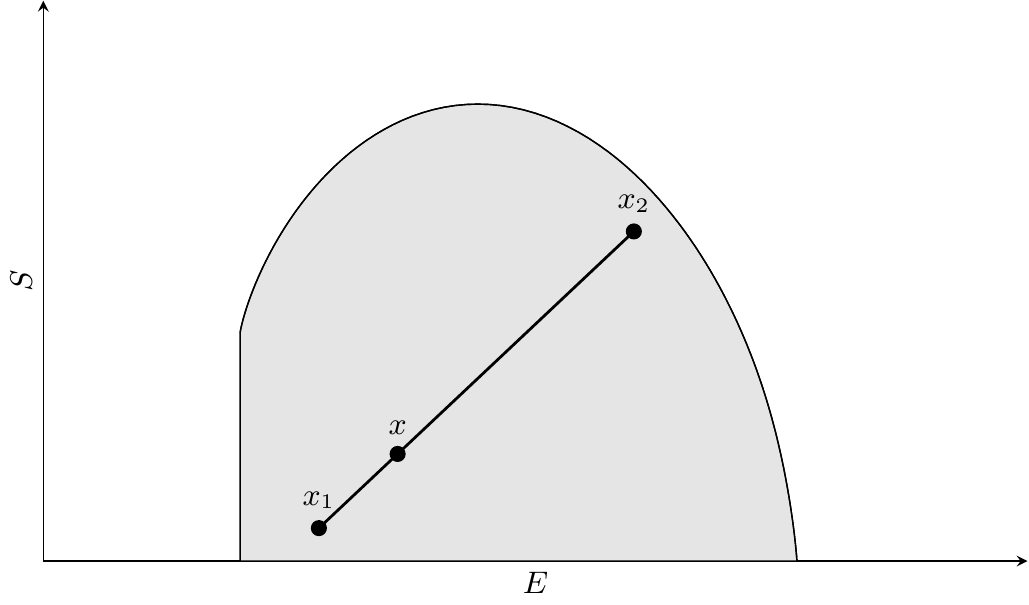}
\caption{Combining two systems in normalised macrostates $x_1$ and $x_2$ results in a total system in the normalised macrostate $x$, which is a convex combination of $x_1$ and $x_2$, where the coefficient of $x_1$---proportional to the distance between $x$ and $x_2$---is equal to the size of the first system relative to the total system, and similarly for $x_2$.}
\label{comb_fig}
\end{figure}
This implies that everything that we do with the convex cone can alternatively be done directly in the
energy-entropy diagram, by simply normalising the macrostates and keeping track of system size separately.

At the purely mathematical level, all of this is nicely analogous to the issue of normalisation of density matrices: it is usually more intuitive to assume the normalisation, and therefore one often normalises explicitly; but it is occasionally also advantageous to use unnormalised density matrices in order to keep track of the normalisation, which represents a ``probability-to-occur'', analogous to our system size coordinate. Conversely, it is often useful to decompose a given normalised density matrix into a convex combination of other ones, such as pure states; many of the puzzling features of quantum theory can be attributed to the fact that such a decomposition is highly non-unique\footnote{In the sense that two decomposition do in general not have a common refinement.}. The same applies to thermodynamics: it may occasionally be useful to write a normalised macrostate $x$ as a convex combination of other ones, or equivalently to decompose a given $y\in\Therm(H)$ into a sum $y=y_1 + y_2$ for $y_i\in\Therm(H)$. Of particular interest are decompositions into normalised macrostates that are extreme points of the energy-entropy diagram. Again such decompositions are highly non-unique, and we will argue that this non-uniqueness is among the essential features of thermodynamics and underlies e.g.~the possibility of constructing heat engines (Section~\ref{heat_engines}). What this means is that macroscopic thermodynamics is, purely mathematically, an example of a \emph{general probabilistic theory}~\cite{barrett_gpt_2007,barnum_entropy_2015,chiribella_entanglement_2015}. The physical meaning, however, is very different from how one usually thinks of a general probabilistic theory such as quantum theory, and in fact already applies at the level of \emph{classical} thermodynamics, and in fact already applies at the level of \emph{classical} thermodynamics.
\par
In more detail, Proposition~\ref{ESconv} implies that the extreme points of the energy-entropy diagram are the following:
\begin{itemize}
	\item The thermal macrostates $x(\tau_\beta)$, for $\beta\in[-\infty,+\infty]$, which are all different unless $H=0$.
	\item The pure macrostate $x(|E_{\min}\rangle\langle E_{\min}|)$, which coincides with the ground state $x(\tau_\infty)$ in the case of non-degeneracy, and the pure macrostate $x(|E_{\max}\rangle\langle E_{\max}|)$, which may similarly coincide with the maximally excited state $x(\tau_{-\infty})$.
\begin{figure}[!ht]
	\centering
	\includegraphics[width=1\columnwidth]{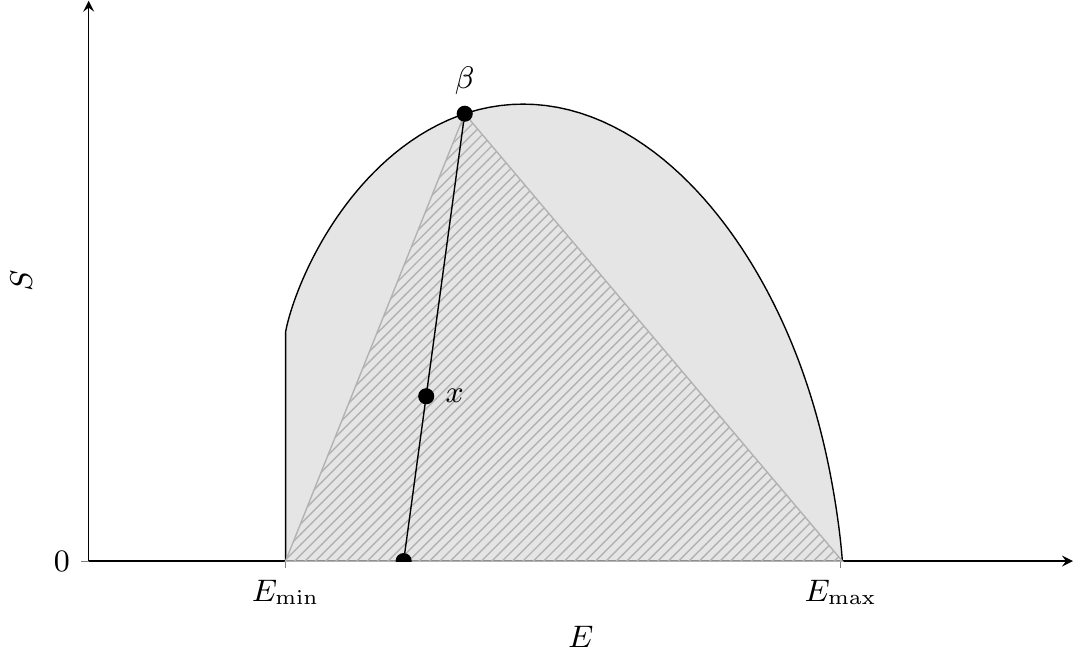}
	\caption{Decomposing a normalised macrostate into a combination of the thermal macrostate at temperature $\beta^{-1}$ and a pure macrostate, together with the set of all states that have such a decomposition at the given $\beta$ (hatched).}
	\label{convdecomp}
\end{figure}
\end{itemize}
By first writing every normalised macrostate as a convex combination of a thermal state and a pure macrostate (Figure~\ref{convdecomp}), and then decomposing the pure macrostate further into a combination of ground states and maximally excited states, we can represent every normalised macrostate as a convex combination of a suitable thermal state $x(\tau_\beta)$ with $x(|E_{\min}\rangle\langle E_{\min}|)$ and $x(|E_{\max}\rangle\langle E_{\max}|)$. For a given $\beta$, Figure~\ref{convdecomp} shows the region of normalised macrostates $x$ that have a decomposition of this form for fixed $\beta$. So we can write any $y(\rho) \in \Therm(H)$ as
\begin{align}
\frac{y(\rho)}{n(\rho)} &= c_\beta \cdot x(\tau_\beta) \nonumber \\
&+ c_{\min} \cdot x(|E_{\min}\rangle\langle E_{\min}|)  \\[4pt]
&+ c_{\max} \cdot x(|E_{\max}\rangle\langle E_{\max}|),\nonumber
\end{align}
for suitable weights $c_\beta,c_{\min},c_{\max}\in [0,1]$. If we choose rational approximations for these coefficients and suitably rescale the system size such that the product of each coefficient with $n(\rho)$ is an integer, then we even obtain an asymptotic equivalence,
\begin{align}
	\rho &\:\asmpequiv\: \tau_\beta^{\otimes c_\beta n(\rho)} \nonumber \\
	&\otimes |E_{\min}\rangle\langle E_{\min}|^{\otimes c_{\min}n(\rho)}  \\[3pt]
	&\otimes |E_{\max}\rangle\langle E_{\max}|^{\otimes c_{\max}n(\rho)}. \nonumber
\end{align}
In this way, any state looks macroscopically like a combination of a number of thermal states (at one temperature), ground states, and maximally excited states. The non-uniqueness in this decomposition lies in the possibility of choosing the temperature $\beta^{-1}$; the number of states of each kind in the decomposition will vary with that temperature. In practice, one can fix $\beta$ first and then determine the coefficients by equating energy, entropy and system size of the two sides of the equation. If these coefficients turn out to be nonnegative, then one has found a feasible decomposition. This is an instance of decomposing a state in a general probabilistic theory into pure states.

In summary, at the mathematical level, macroscopic thermodynamics is a general probabilistic theory. Some of the essential features of thermodynamics are intimately related to the non-uniqueness of decompositions of states into extremal states---the same phenomenon that is behind many of the mysterious aspects of quantum theory. But although the mathematics is an instance of the formalism of general probabilistic theories, the physical meaning is very different from the standard interpretation of the latter~\cite{barrett_gpt_2007}.

\section{The maximal extractable work}
\label{max_work}

We now use the methods developed in the previous sections to analyze the maximal work that can be extracted from many copies of a given state $\rho$. Following the standard definition of \emph{work}, we introduce a separate subsystem called \emph{battery} consisting of $\ell$ copies of the pure ground state,
\[
\omega_{\text{in}} = \rho^{\otimes n} \otimes \ket{E_{\min}}\bra{E_{\min}}^{\otimes \ell},
\]
where $\ell$ is yet to be determined. Extracting work from $\rho^{\otimes n}$ means that one devises a protocol which turns $\omega_{\text{in}}$ into a final state in which the battery system is in the maximally excited state,
\[
\omega_{\text{fin}} = \sigma^{\otimes n} \otimes \ket{E_{\max}}\bra{E_{\max}}^{\otimes \ell},
\]
where $\sigma$ is a state to be determined in such a way that the work, which is the energy exchanged with the battery per copy of $\rho$,
\[
W = \frac{\ell}{n}(E_{\max} - E_{\min}) = E(\rho) - E(\sigma)
\]
is maximised. By Theorem~\ref{TheoSEmain}, such a protocol exists if and only if $\omega_{\text{in}}$ and $\omega_{\text{fin}}$ have the same energy and entropy. Since the battery subsystem has no entropy, this means that we need to have $S(\rho) = S(\sigma)$; and moreover, in order to maximise $\ell$, we should choose $\sigma$ such that its energy is minimised. Having minimal energy for given entropy $S(\rho)$ means that we draw the horizontal line in the energy-entropy diagram through the macrostate $x(\rho)$, and determine the point where it hits the boundary; then we need to choose $\sigma$ to be a representative of this macrostate. In the regime where $S(\rho)$ is greater than the ground state degeneracy, we therefore have $\sigma = \tau_{\tilde{\beta}}$, where the
temperature $\tilde{\beta}^{-1} > 0$ is such that $S(\tau_{\tilde{\beta}}) = S(\rho)$. In this case, the maximal amount of work that we can extract asymptotically per copy of $\rho$ is given by
\begin{equation}
W_{\text{max}}(\rho) = E(\rho) - E(\tau_{\tilde{\beta}}),
\end{equation}
in agreement with existing results~\cite{alicki_entanglement_2013,allahverdyan_maximal_2004}.

In conclusion, the maximal total amount of pure energy that can be extracted from $\rho^{\otimes n}$ can be neatly read off the energy-entropy diagram: it is given by the horizonal distance between the macrostate $x(\rho)$ and the curve of thermal states,
multiplied by the system size $n$.

\section{Work and heat}
\label{work_heat_sec}

As we have seen in the previous section, coupling a system to a battery allows for the exchange of energy with the battery. So in contrast to the situation of Theorem~\ref{TheoSEmain}, this means that now only entropy is a conserved quantity, and one can move between any two macrostates which are on the same horizontal line in the energy-entropy diagram.

So what can we do in order to move between states that are not even on the same horizontal line? In analogy with adjoining a battery, it is natural to do this by adding a thermal reservoir of finite size, with which the original system can then exchange \emph{heat}. So while energy exchanged with the battery is what we consider work, energy exchanged with the thermal reservoir is our definition of heat. Assuming an environment to consist of a reservoir plus a battery is also motivated by the decomposition of any macrostate into a thermal and a pure part as in the previous section.

Again using our previous results, we will compute the work and heat required to transform any given state $\rho$ into any desired state $\sigma$, as a function of the initial and final temperature of the reservoir. When the size of the reservoir tends to infinity, or equivalently when its initial and final temperature coincide, the work and heat exchanged specialise to the standard ones in terms of the free energy.
\par
Getting to the technicalities, each subsystem consists of any number of microsystems with Hamiltonian $H$ as before. We assume that the thermal reservoir consists initially of $m$ copies of the thermal state at some temperature $\beta^{-1}_1$, and the battery of $\ell$ copies of a pure ground state. The initial state of the total system is therefore given by
\begin{equation} \label{init_state_work_heat}
\omega_{\text{in}} = \rho^{\otimes n} \otimes \tau_{\beta_1}^{\otimes m} \otimes \ket{E_{\min}}\bra{E_{\min}}^{\otimes \ell}.
\end{equation}
Since $\rho$ and $\sigma$ may have different entropy and average energy, turning the former
into the latter means that we also have to modify the reservoir and battery states. In
fact, in order to convert $\rho$ into $\sigma$, we apply Theorem~\ref{TheoSEmain} to the final state of all three subsystems,
which we assume to be close to
\begin{equation} \label{final_state_work_heat}
\omega_{\text{fin}} = \sigma^{\otimes n} \otimes \tau_{\beta_2}^{\otimes m} \otimes \ket{E_{\max}}\bra{E_{\max}}^{\otimes \ell},
\end{equation}
where now the reservoir is in a thermal state at a possibly different temperature $\beta_2^{-1}$, and
the battery is in the maximally excited state.
\par
It is worth noting that, in general, the final state of the reservoir does not have to be thermal, since the
interaction with the system might have driven the environment out of equilibrium. However, if the final
state of the reservoir is athermal, one would be able to extract additional work from it, while keeping
the entropy of this system unchanged (as we have previously shown in Sec.~\ref{max_work}).
According to our definition of a battery (as a work-exchanging device), we have that all the possible
work associated with the system transformation $\rho \rightarrow \sigma$ should be exchanged with it,
and none should be locked inside the battery. For this reason, it seems natural to ask the final state of the
reservoir to be thermal.
\par
Now in order for Theorem~\ref{TheoSEmain} to apply, we need to consider the asymptotic limit, that is,
when $n$, $m$, and $\ell \gg 1$. In this case, we can convert $\omega_{\text{in}}$ into $\omega_{\text{fin}}$
using the set of allowed operations if and only if they have the same average energy and entropy; conservation of system size is already guaranteed to hold. This gives two equations that we can solve for $m$ and $\ell$, resulting in
\begin{equation} \label{number_states}
\frac{m}{n} = \frac{S(\sigma) - S(\rho)}{S(\tau_{\beta_1}) - S(\tau_{\beta_2})},
\end{equation}
and a somewhat more complicated expression for $\tfrac{\ell}{n}$. So in order for $m$ to be nonnegative, we should have $\beta_1 < \beta_2$ if $S(\rho) > S(\sigma)$ and vice versa (assuming that $\beta_1,\beta_2 > 0$). Physically, this implies that when we dump entropy
from the system into the thermal reservoir, we increase its temperature, and vice versa, as we would
expect in the case of a finite size thermal reservoir. We refer to our thermal reservoir as being of finite size because, even if it is composed of $m\to\infty$ copies, the reservoir size is finite \emph{relative to} the system size $n$, in contrast to the case analysed in~\cite{brandao_resource_2013}. Similarly, a positive value $\ell > 0$ means that we achieve an extraction of work from the system; while if $\ell$ comes out negative, then we can make it positive by taking the initial state of the battery to be $\ket{E_{\max}}^{\otimes\ell}$, and the final one $\ket{E_{\min}}^{\otimes \ell}$, which corresponds to an injection  of work into the system. For simplicity, we focus on the case that $m,\ell > 0$ with $\omega_{\text{in}}$ and $\omega_{\text{fin}}$ as above, while the other cases are analogous.
\par
We can now evaluate the work extracted and heat provided during the state transformation.
We identify these two quantities with, respectively, the energy difference between the final and
initial $\ell$ copies of pure states, and with the energy difference between the initial and final $m$ copies
of thermal states. Thus, work is the energy stored inside the pure states, and heat is the energy
exchanged with the thermal states. Using the result of Eq.~\eqref{number_states}, we obtain the
following expressions for the work extracted and the heat provided per copy of $\rho$
and $\sigma$,
\begin{align}
\label{work_def}
W_{\beta_1, \beta_2}(\rho \rightarrow \sigma) &=
\frac{\ell}{n} \, \left( E_{\max} - E_{\min} \right) \nonumber \\
&=\left( E(\rho) - E(\sigma) \right) \\
&- \frac{E(\tau_{\beta_1}) - E(\tau_{\beta_2})}{S(\tau_{\beta_1}) - S(\tau_{\beta_2})} \left( S(\rho) - S(\sigma) \right),\nonumber \\[5pt]
\label{heat_def}
Q_{\beta_1, \beta_2}(\rho \rightarrow \sigma) &=
\frac{m}{n} \, \left( E(\tau_{\beta_1}) - E(\tau_{\beta_2}) \right) \nonumber \\
&= \frac{E(\tau_{\beta_1}) - E(\tau_{\beta_2})}{S(\tau_{\beta_1}) - S(\tau_{\beta_2})} \left( S(\sigma) - S(\rho) \right).
\end{align}
These quantities depend on the initial and final system state, but also on the initial and final temperature of the reservoir. Our definition of work and heat is consistent with the
first law of thermodynamics, since we have $\Delta E(\rho \rightarrow \sigma) =
Q_{\beta_1, \beta_2}(\rho \rightarrow \sigma) - W_{\beta_1, \beta_2}(\rho \rightarrow \sigma)$,
where $\Delta E = E(\rho) - E(\sigma)$ is the average energy difference between the final and initial state of the
system, independently of $\beta_1$ and $\beta_2$.
\par
These equations for work and heat are similar to the standard ones. In fact, work is given by the free energy difference between $\rho$ and $\sigma$,
for an external effective temperature $\beta_{\text{eff}}^{-1}$ depending on the initial and
final temperatures of the thermal reservoir,
\begin{equation} \label{beta_eff}
\beta_{\text{eff}}(\beta_1, \beta_2) =  \frac{S(\tau_{\beta_1}) - S(\tau_{\beta_2})}
{E(\tau_{\beta_1}) - E(\tau_{\beta_2})},
\end{equation}
so that
\begin{equation}
\label{final_work_def}
W_{\beta_1, \beta_2}(\rho \rightarrow \sigma) = \beta_{\text{eff}}^{-1}
\left( A_{\beta_{\text{eff}}}(\rho) - A_{\beta_{\text{eff}}}(\sigma) \right).
\end{equation}
In the same way, the equation for heat is equal to the standard one, for the same
effective temperature $\beta_{\text{eff}}^{-1}$,
\begin{equation}
\label{final_heat_def}
Q_{\beta_1, \beta_2}(\rho \rightarrow \sigma) = \beta_{\text{eff}}^{-1}
\left( S(\sigma) - S(\rho) \right).
\end{equation}
This equation can also be seen as a non-infinitesimal generalisation of the fundamental thermodynamic relation $dQ = \beta^{-1} dS$. In other words, by the defining Eq.~\eqref{beta_eff}, the effective temperature $\beta_{\text{eff}}^{-1}$ can be visualised as a slope in the energy-entropy diagram, as in Figure~\ref{beta_eff_fig}.
\begin{figure}[!ht]
\centering
\includegraphics[width=1\columnwidth]{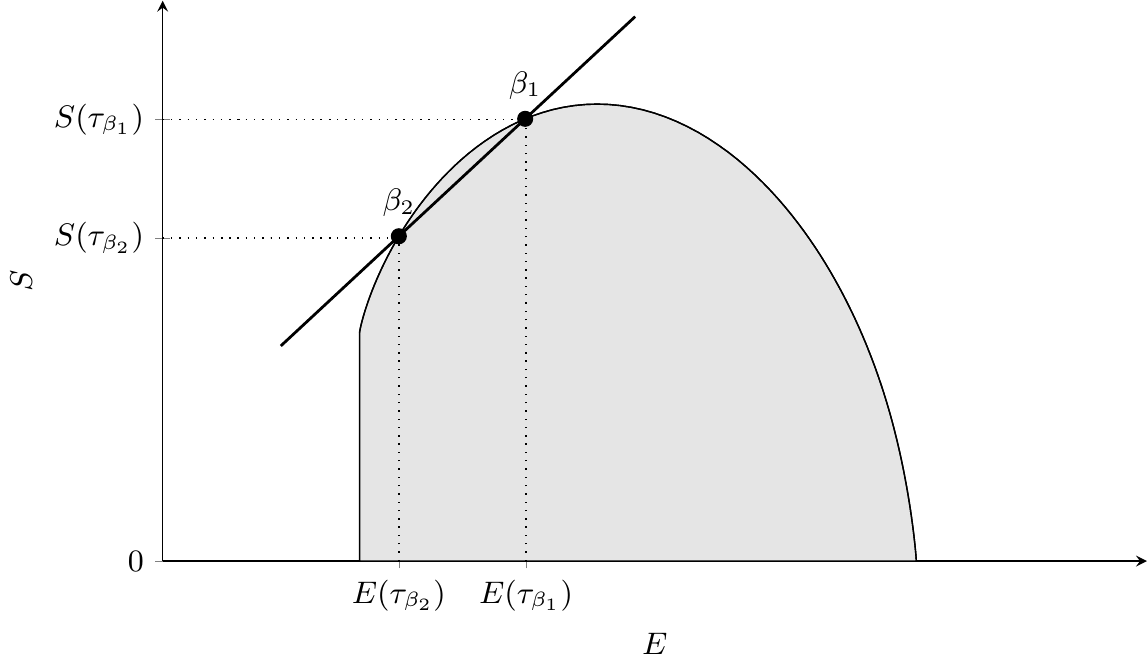}
\caption{The visualisation of the effective inverse temperature of the reservoir, $\beta_{\text{eff}}$.
The thermal reservoir initially consists of $m$ copies of $\tau_{\beta_1}$, which are each turned into
$\tau_{\beta_2}$ by the state transformation.
 The value of $\beta_{\text{eff}}$ is given by the slope of the line connecting the two corresponding points in the energy-entropy diagram. When $\beta_2 = \beta_1 + \varepsilon$, for $| \varepsilon | \rightarrow 0$, the two points get
closer and closer, and the line approaches the tangent to the curve of thermal states. In this case,
$\beta_{\text{eff}} = \beta_1 + O(\varepsilon)$, by Remark~\ref{sloperemark}.}
\label{beta_eff_fig}
\end{figure}
\par
Let us now consider the limiting case of an infinite reservoir, so that the reservoir temperature changes only infinitesimally.
In this case, we have $\beta_2 = \beta_1 + \varepsilon$, where $| \varepsilon | \ll 1$.
Then, it is straightforward to show that $\beta_{\text{eff}} = \beta_1 + O(\varepsilon)$,
and the work and heat we obtain are equal to the standard ones (up to first order in
$\varepsilon$), that is,
\begin{align*}
W_{\text{standard}}(\rho \rightarrow \sigma) & = \beta_1^{-1}
\left( A_{\beta_1}(\rho) - A_{\beta_1}(\sigma) \right) + O(\varepsilon),\\[4pt]
Q_{\text{standard}}(\rho \rightarrow \sigma) & = \beta_1^{-1} \left( S(\sigma) -
S(\rho) \right) + O(\varepsilon).
\end{align*}
Moreover, we find from Eq.~\eqref{number_states}
that, when we \emph{want} the temperature change to be only $\eps\ll 1$, then the required size of the thermal reservoir
per copy of the system $S$ tends to infinity, according to
\begin{equation}
\frac{m}{n} = \frac{S(\sigma) - S(\rho)}{\beta_1 \langle \Delta^2 H \rangle_{\tau_{\beta_1}}}
\frac{1}{\varepsilon} + O(1),
\end{equation}
where the expectation value in the denominator is the variance of energy in the state $\tau_{\beta_1}$,
or equivalently $\beta_1^{-2}$ times the heat capacity (at $\beta_1$) of a single system.
\par
We close this section with an application of our formalism, relevant for classical and quantum computation, and
information processing tasks. We consider the erasure of information, or Landauer's erasure, in the
scenario in which the surrounding environment has a finite size~\cite{reeb_improved_2014}.
In this scenario, the main system is acting as a memory, and its energy does not change during the
transformation, meaning that $E(\rho) = E(\sigma)$. Then we find that, when the thermal reservoir
has a finite size, the work required to erase the state $\rho$, and to map it into $\sigma$, is
\begin{equation}
W^{\text{erasure}}_{\beta_1, \beta_2}(\rho \rightarrow \sigma) =
\beta^{-1}_{\text{eff}}(\beta_1, \beta_2) \left( S(\rho) - S(\sigma) \right),
\end{equation}
which converges to the well-known value $\beta^{-1}_1 \left( S(\rho) - S(\sigma) \right)$ when the size of
the reservoir tends to infinity. If the initial state $\rho$ is maximally-mixed, and the final state
$\sigma$ is pure, we find that the work of erasure is $\beta^{-1}_{\text{eff}}(\beta_1, \beta_2) \log d$.

\section{Heat Engines}
\label{heat_engines}

We now show how the results of the previous sections can be used in order to analyse the efficiency of heat engines (and refrigerators) utilising finite size reservoirs. We do not assume any specific kind of engine consisting of a particular device or using particular mechanisms; instead, we utilise our formalism in order to derive the maximal efficiency of \emph{any} protocol operating on two finite size reservoirs. As before, our analysis is valid in the limit of many copies and in the case where all systems consist of (approximately) non-interacting microsystems with common Hamiltonian $H$. For example, we can now see how a heat engine secretly exploits our observation that thermodynamics is a general probabilistic theory (Section~\ref{decomp}): if we have a system consisting of two subsystems given by thermal states at different temperatures, we start with a macrostate which is a convex combinations of two extreme points of the energy-entropy diagram. Decomposing it in a different way into extreme points, we can therefore extract a certain number of copies of the maximally excited state---which plays the role of extracted work---together with a thermal state at an intermediate temperature.

For heat engines and refrigerators utilising reservoirs of finite size, we derive explicit expressions for the maximal efficiency depending on two effective temperatures (describing the hot and cold reservoirs, respectively). As per Eqs.~\eqref{eng_heat} and \eqref{eng_ref}, this optimal efficiency with finite-size reservoirs is always lower than the Carnot efficiency.
\par
Our model consists of the same tripartite system as in the previous section, but further specialised to the case where both the initial state $\rho$ and the final state $\sigma$ are themselves thermal. Hence the initial state is given by
\begin{equation} \label{init_state_engine}
\omega^{\text{engine}}_{\text{in}} = \tau_{\beta_\text{cold}}^{\otimes n} \otimes
\tau_{\beta_\text{hot}}^{\otimes m} \otimes \ket{E_{\min}}\bra{E_{\min}}^{\otimes \ell},
\end{equation}
where $\beta_\text{cold} > \beta_\text{hot}$. The final state, instead, is
\begin{equation} \label{final_state_engine}
\omega^{\text{engine}}_{\text{fin}} = \tau_{\beta_\text{less-cold}}^{\otimes n} \otimes
\tau_{\beta_\text{less-hot}}^{\otimes m} \otimes \ket{E_{\max}}\bra{E_{\max}}^{\otimes \ell},
\end{equation}
where $\beta_\text{cold} > \beta_\text{less-cold} > \beta_\text{less-hot} > \beta_\text{hot}$.
The engine uses the hot and cold reservoirs to extract work, but in the meanwhile it degrades
these reservoirs, assimilating their temperatures (because these are of finite size).
\par
Since everything that we do is reversible, one can consider both the transformation $\omega^{\text{engine}}_{\text{in}} \rightarrow \omega^{\text{engine}}_{\text{fin}}$ (heat engine) as well as the reverse $\omega^{\text{engine}}_{\text{fin}} \rightarrow \omega^{\text{engine}}_{\text{in}}$ (refrigerator). We are not concerned with the question of how to realise these transformations; Theorem~\ref{TheoSEmain} gives us necessary and sufficient conditions for when they are realisable, but does not make any statement about how to implement them, using a ``working body'' or otherwise. We only know that there exists some unitary acting on the global system together with a small number of ancilla systems which realises these devices to any desired degree of accuracy as $m,n,\ell\to\infty$. And moreover, there is no other device or mechanism that could do better.
\par
In order to evaluate the efficiency of these two devices, we need to evaluate the heat exchanged
with the hot reservoir, the work produced or utilised, and the heat exchanged with the cold reservoir. Due to reversibility, these quantities are the same for both devices (at least in absolute value). Using Eqs.~\eqref{work_def} and~\eqref{heat_def}, we find the heat exchanged with the hot reservoir $Q_{\text{hot}}$, and the work exchanged $W$,
\begin{align*}
Q_{\text{hot}} &= \beta^{-1}_{\text{eff}}(\beta_\text{hot}, \beta_\text{less-hot})
\big( S(\tau_{\beta_\text{less-cold}}) - S(\tau_{\beta_\text{cold}}) \big), \\
W &= \big( E(\tau_{\beta_\text{cold}}) - E(\tau_{\beta_\text{less-cold}}) \big) \\
&- \beta^{-1}_{\text{eff}}(\beta_\text{hot}, \beta_\text{less-hot})
\big( S(\tau_{\beta_\text{cold}}) - S(\tau_{\beta_\text{less-cold}}) \big),
\end{align*}
both per copy of the first reservoir system. On the other hand, since the system $S$
is now the cold reservoir, the heat $Q_{\text{cold}}$ exchanged with it per copy is equal to
\[
Q_{\text{cold}} = E(\tau_{\beta_\text{less-cold}}) - E(\tau_{\beta_\text{cold}}).
\]
\par
We can now evaluate the  efficiency of the heat engine, defined as $\eta_{\text{engine}}
= \frac{W}{Q_{\text{hot}}}$, and the efficiency of the refrigerator, $\eta_{\text{refrigerator}}
= \frac{Q_{\text{cold}}}{W}$. We find that the efficiencies are equal to
\begin{align}
\label{eng_heat}
\eta_{\text{engine}} &= 1 -
\frac{\beta_{\text{eff}}(\beta_\text{hot}, \beta_\text{less-hot})}
{\beta_{\text{eff}}(\beta_\text{cold}, \beta_\text{less-cold})},\\
\label{eng_ref}
\eta_{\text{refrigerator}} &= \left(
\frac{\beta_{\text{eff}}(\beta_\text{cold}, \beta_\text{less-cold})}
{\beta_{\text{eff}}(\beta_\text{hot}, \beta_\text{less-hot})} - 1 \right)^{-1},
\end{align}
where $\beta_{\text{eff}}$ was defined in Eq.~(\ref{beta_eff}). 
\begin{figure}[!ht]
\centering
\includegraphics[width=1\columnwidth]{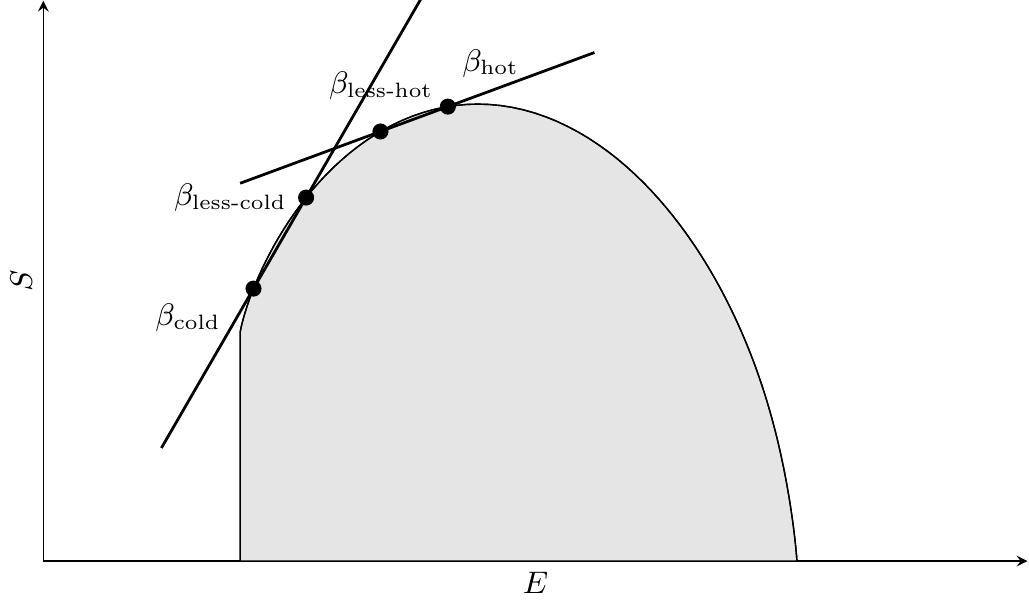}
\caption{The two effective temperatures that determine engine efficiency interpreted as slopes in the energy-entropy diagram. For $\beta_{\text{less-cold}}$ very close to $\beta_{\text{cold}}$ and $\beta_{\text{less-hot}}$ very close to $\beta_{\text{hot}}$, i.e.~when the reservoirs are very large compared to the battery size, then the lines approximate tangents and the resulting efficiency approaches the Carnot efficiency.}
\label{beta_eff2_fig}
\end{figure}
In terms of the interpretation of effective inverse temperatures as slopes in the energy-entropy diagram, we can understand these efficiencies as in Figure~\ref{beta_eff2_fig}: for example for the heat engine, the quotient of the slopes is always less than the quotient of the two tangent slopes at $\beta_{\text{hot}}$ and $\beta_{\text{cold}}$, respectively. This implies that $\eta_{\text{engine}} < 1 - \frac{\beta_\text{hot}}{\beta_\text{cold}}$, and similarly $\eta_{\text{refrigerator}} < \left( \frac{\beta_\text{less-cold}}{\beta_\text{less-hot}} - 1 \right)^{-1}$, so that both efficiencies are strictly lower than the Carnot efficiencies. This is due to the fact that the temperature of the two finite size reservoirs changes during the process. In the limit where the temperature of the two reservoirs changes only by an infinitesimal amount, both efficiencies approach Carnot's values.

\section{Optimal rates of conversion and how to compute them}
\label{opt_rates}

So far, we have only considered asymptotic equivalence of states, since only energy-preserving unitaries have been allowed. What do we get if we allow in addition that subsystems can be discarded? Building on our previous results, we will now give one possible answer to this question. In all cases, we assume that a state $\rho$ lives on a certain number $n(\rho)$ of microsystems as before, and when we are dealing with two (or more) states $\rho$ and $\sigma$, we do not assume $n(\rho) = n(\sigma)$.

\newcommand{\sub}{\mathrm{sub}}

\begin{lemma}[The asymptotic ordering $\succeq$]
	For given states $\rho$ and $\sigma$, the following are equivalent:
	\begin{enumerate}
		\item\label{tensordecomp} There is a state $\phi$ such that $\rho\asmpequiv\sigma\otimes\phi$ in the sense of Theorem~\ref{TheoSEmain}.
		\item\label{subsigma} There is a state $\sigma'\asmpequiv\rho$ such that $\Tr{\sub}{\sigma'} = \sigma$ for some subsystem that is not entangled with the rest.
		\item\label{SForder} $S(\rho) \geq S(\sigma)$ and $A_\beta(\rho) \geq A_\beta(\sigma)$ for all $\beta\in(-\infty,+\infty)$.
		\item\label{coneorder} We have $y(\rho) - y(\sigma)\in \Therm(H)$.
	\end{enumerate}
	\label{discardequiv}
\end{lemma}

Let us write $\rho \succeq \sigma$ for the ordering relation on states corresponding to these equivalent conditions. In the following, we will investigate this ordering relation a bit further.

\begin{proof}
	The implication from~\ref{tensordecomp} to~\ref{subsigma} is simply by taking $\sigma':=\sigma\otimes\phi$.

	Next, we show that~\ref{subsigma} implies~\ref{SForder}; this follows from the fact that the additive monotones $S$ and $A_\beta$ respect asymptotic equivalence $\asmpequiv$, and are nonincreasing under tracing out such subsystems. The former is a consequence of Theorem~\ref{TheoSEmain}, while the latter is a consequence of the no-entanglement assumption in the case of $S$ (the conditional entropy is nonnegative) and of the data processing inequality in the case of $A_\beta$,
	\begin{align}
		A_\beta(\rho) &= D(\rho\|\tau_\beta^{\otimes n(\rho)})
		\geq D\left(\Tr{\sub}{\rho}\Big\|\Tr{\sub}{\tau_\beta^{\otimes n(\rho)}}\right) \nonumber \\
		&= A_\beta(\Tr{\sub}{\rho}),
	\end{align}
	where the last equation holds because $\Tr{\sub}{\tau_\beta}$ is the thermal state on the subsystem.

	From~\ref{SForder} to~\ref{coneorder}, the point $y(\rho) - y(\sigma)$ satisfies all the inequalities of Proposition~\ref{conechar} by assumption, and therefore lies in $\Therm(H)$.

	From~\ref{coneorder} to~\ref{tensordecomp}, the assumption together with Remark~\ref{slicetherm} guarantees the existence of a state $\phi$ with $y(\phi) = y(\rho) - y(\sigma)$. Therefore $y(\rho) = y(\sigma\otimes\phi)$, and then $\rho\asmpequiv \sigma\otimes\phi$ is a consequence of Theorem~\ref{TheoSEmain}.
\end{proof}

The ``no entanglement'' requirement in condition~\ref{subsigma} seems a bit artificial, and it would be interesting to obtain results analogous to the upcoming ones for the ordering relation defined in the analogous way, but where one would be allowed to trace out an \emph{arbitrary} subsystem. We suspect that such a development would require generalisations of Theorems~\ref{Sthm} and~\ref{TheoSEmain}, where instead of characterising the asymptotic equivalence of states relative to energy-preserving unitaries, one would instead classify the asymptotic \emph{ordering} of states relative to energy-preserving unitaries and discarding subsystems. We currently do not have such a result and thus use the $\succeq$ relation from Lemma~\ref{discardequiv}.

\begin{definition}[{\cite[Eq.~(8.3)]{fritz_resource_2015}}]
	The \emph{maximal rate} of converting a state $\rho$ into a state $\sigma$ is given by
	\beqn
		\label{ratedef}
		R_{\max}(\rho\to\sigma) := \sup\:\left\{\: \frac{m}{n} \:\bigg|\: \rho^{\otimes n} \succeq \sigma^{\otimes m} \:\right\}
	\eeqn
\end{definition}

So roughly speaking, we now ask: if we try to convert many copies of $\rho$ into many copies of $\sigma$, then how many copies of $\rho$ do we need per copy of $\sigma$, where we may discard some additional ``junk'' states in the process?

Since we already have allowed for sublinear ancillas in the definition of asymptotic equivalence, this notion of maximal rate actually corresponds to the notion of \emph{regularised} maximal rate of~\cite[Section~8]{fritz_resource_2015}. Building on the methods that we have developed so far, it is not hard to write down a concrete formula for computing maximal rates:

\begin{theorem}
	The maximal rate from $\rho$ to $\sigma$ can be computed in two ways:
	\begin{enumerate}
		\item $R_{\max}(\rho\to\sigma)$ is equal to the value of $r$ at which the line in $\R^3$ defined by $r\mapsto y(\rho) - r y(\sigma)$ pinches the boundary of the cone $\Therm(H)$, so that
			\begin{align}
				R_{\max}(\rho\to\sigma) = \max \big\{\: &r \in \R_{\geq 0} \: \big| \nonumber \\
				&y(\rho) - r y(\sigma) \in \Therm(H) \:\big\}. \label{ratesup}
			\end{align}
		\item $R_{\max}(\rho\to\sigma)$ is also equal to the minimal ratio of the value of an additive monotone on $\rho$ versus its value on $\sigma$,
			\beqn
			\label{rateinf}
				R_{\max}(\rho\to\sigma) = \min\left\{ \frac{S(\rho)}{S(\sigma)}, \inf_{\beta\in (-\infty,+\infty)} \frac{A_\beta(\rho)}{A_\beta(\sigma)} \right\},
			\eeqn
			where the minimization is only over those fractions for which the denominator is nonzero.
	\end{enumerate}
	\label{ratethm}
\end{theorem}

\begin{proof}
	By additivity of $y$, a rational number $r=\tfrac{p}{q}\in\Q_{\geq 0}$ is an achievable rate if and only if $q y(\rho) - p y(\sigma) \in \Therm(H)$, or equivalently $y(\rho) - r y(\sigma)\in\Therm(H)$. This implies~\eqref{ratesup}.

	One gets~\eqref{rateinf} from~\eqref{ratesup} via Proposition~\ref{conechar}, since $y(\rho) - r y(\sigma)\in\Therm(H)$ is equivalent to $S(\rho) \geq r S(\sigma)$ together with $A_\beta(\rho) \geq r A_\beta(\sigma)$ for all $\beta\in(-\infty,+\infty)$. Hence the condition on $r$ is that it must be less than or equal to $\tfrac{S(\rho)}{S(\sigma)}$, and also less than or equal to $\tfrac{A_\beta(\rho)}{A_\beta(\sigma)}$ for every $\beta$, for those fractions for which the denominator is nonzero. The largest $r$ that satisfies this is precisely~\eqref{rateinf}.
\end{proof}

To understand Eq.~\eqref{ratesup} intuitively, it may help to normalise the macrostates and phrase the condition in terms of convex combinations in the energy-entropy diagram instead, as per Section~\ref{decomp}.

What makes the infimum over $\beta$ in Eq.~\eqref{rateinf} nontrivial to evaluate is the presence of the partition function term $\log Z_\beta$ in both the numerator and the denominator, due to Eq.~\eqref{Abeta}. Nevertheless, this is a very explicit formula with which one should be able to compute rates in practice. It is an instance of~\cite[Theorem~8.24]{fritz_resource_2015}, and the proof is correspondingly similar.
\section{Conclusion}
\label{concl}
Our resource theory for thermodynamics does not make use of an infinite
thermal reservoir. Therefore, it is suitable for analysing state transformations
both when the system is decoupled from the environment, e.g.~via Eq.~\eqref{rateinf},
and when the system is interacting with a finite reservoir, Eqs.~\eqref{final_work_def} and \eqref{final_heat_def}.
Moreover, the theory provides a rigorous mathematical explanation (through the
Thm.~\ref{TheoSEmain}) to the fact that, when dealing with
macroscopic thermodynamics, we can describe the state of a system with few
observables (for instance, energy and entropy). Our approach generalises the one
presented in \cite{brandao_resource_2013}, where asymptotic state
transformations are considered when an infinite reservoir is present.
\par
The results we obtain are valid in a specific regime delineated by several idealised assumptions,
such as the assumption that all energy-preserving unitaries are available, the presence of
many non-interacting and identical copies of the system, and the constraint of a fixed Hamiltonian
for each system. One can think of dropping some of these assumptions, and for example investigate
the theory when arbitrary states and interactions are allowed (often called the single-shot regime),
or when one has a much more realistic class of operations not requiring such fine grained control of
system and bath~\cite{perry_sufficient_2015}.
\par
The asymptotic equivalence result presented in Eq.~\eqref{asympt_eq} is obtained with the help
of a sublinear ancillary system. A priori, one might think that this additional system could be used
as an unbounded source of work, since we do not require the state of this ancilla to be restored
at the end of the process. However, to avoid the possibility of freely modifying the energy of the system
by exploiting the ancilla, we constrain the energy spectrum of the latter to be sublinear in the number
of copies of the main system.
\par
Recently, resource theories with multiple conserved quantities (even non-commuting ones),
have been investigated within the framework of quantum
thermodynamics~\cite{lostaglio_thermodynamic_2015, guryanova_thermodynamics_2015,
halpern_microcanonical_2015}. However, in these models, emphasis is put on
different notions of work, each of them related to a different conserved quantity. Our theory, on the
other hand, considers only energy\footnote{Or any other \emph{single} conserved quantity---if one
replaces e.g.~``energy'' by ``angular momentum'' throughout our work, our results are still equally
correct and meaningful.} We expect that our approach can be extended
more or less straightforwardly so as to cover multiple commuting conserved quantities as well;
generalising to a treatment of multiple non-commuting conserved quantities may
present new challenges.
\section*{Acknowledgement}
We would like to thank L\'idia del Rio, Joe Renes, Matteo Smerlak and Rob Spekkens for helpful discussions.
CS is supported by the EPSRC Centre for Doctoral Training in Delivering Quantum Technologies. Research
at Perimeter Institute is supported by the Government of Canada through Industry Canada and by the Province
of Ontario through the Ministry of Economic Development and Innovation.
\bibliographystyle{ieeetr}
\bibliography{refthermo}
\appendix
\section{Asymptotic equivalence of states under energy-preserving unitaries}
\label{asy_eq_theorems}
\subsection{Overview}
We show in this section that two states of a quantum system are asymptotically equivalent under
energy-preserving unitaries if and only if they have the same entropy and average energy (Theorem~\ref{TheoSEmain}).
Here, we consider two states asymptotically equivalent when one can turn many identical copies of one
state into the same number of copies of the other state, arbitrarily accurately, by applying a global unitary operation which
preserves energy. The precise statement is in Theorem~\ref{SEthm}.
\par
The main difficulty in our proof is in showing that when two states have same
energy and entropy, then they can be asymptotically mapped one into the other. To prove this implication,
we devise a protocol which converts (many copies of) one state into the other, provided that they have the same
energy and entropy. We now summarise the protocol, in order to provide a simple and physical idea of its mechanism to
the reader. We do this in two cases, one concerning the simpler case of trivial Hamiltonian, and then in general.
\par
When the system has trivial Hamiltonian, we can act on it by means of any unitary operation, and the only assumption
we have about the states $\rho$ and $\sigma$ is that they have the same entropy. Since we work in the asymptotic
regime, where we take the tensor product of many copies of these states, we can use the tools developed in
Shannon theory~\cite{shannon_mathematical_1948, schumacher_quantum_1995, nielsen_quantum_2010}.
In particular, due to the central limit theorem, we can replace the many-copy states $\rho^{\otimes n}$ and
$\sigma^{\otimes n}$ with their typical states, Eqs.~\eqref{rhotyp} and \eqref{sigmatyp}. The use of the typical states highly
simplifies the protocol, since in this way we can divide the Hilbert space into a small number of subspaces with common
properties. State conversion is achieved in the protocol by mapping the probability distribution of the initial typical state
into the one of the final typical state. This is done by introducing an ancillary system with trivial Hamiltonian,
whose size is $O(\sqrt{n \log n})$, in the maximally mixed state. This ancilla provides a source of randomness, and we
modify the probability distribution of the initial state by applying a global unitary operation on system and ancilla, and
tracing out the ancilla. However, a unitary operation can be used only if the transformation is reversible. To assure that
this is the case, another ancillary system is introduced, acting as a register, which allows for a dilation to a unitary. Again, the dimension of this second ancillary
system is $O(\sqrt{n \log n})$, and the Hamiltonian is trivial. The details of the protocol are in the proof of Theorem~\ref{Sthm}.
\par
When the system has non-trivial Hamiltonian, we have to reduce the set of allowed unitary operations to
the sole energy-preserving ones (the ones that commute with the Hamiltonian). With these operations, we have to 
devise a protocol which approximately converts many copies of $\rho$ into $\sigma$ when the two states have same entropy and energy.
The protocol which performs this asymptotic transformation is analogous to the one for trivial Hamiltonian.
The difference is that in this case we have to add an additional ancillary system with non-trivial Hamiltonian,
with which we can exchange both energy and coherence. This ancilla allows us to approximately implement any unitary 
on the system by applying an energy-preserving unitary on both system and ancilla. Due to the constraints on the
energy and entropy of the initial and final state, and to the central limit theorem, we achieve that the size of
this additional ancilla is $O(\sqrt{n \log n})$. Moreover, the spectrum of the ancillary Hamiltonian is bounded
by $O(n^{\frac{2}{3}})$, so that we modify the amount of energy only by a sublinear amount. The details of the
protocol are in the proofs of Lemma~\ref{coherence} and Theorem~\ref{SEthm}.

It is worth noting that
none of the three ancillary systems depends in any way on the state $\rho$ or $\sigma$, meaning that the transformation
can be performed with the same ancillae for any initial and final state.
\subsection{Asymptotic equivalence of quantum states}
Before getting to thermodynamics, it helps to consider an easier case first: two states of a quantum system are many-copies equivalent under unitaries if and only if they have the same entropy. This is what we show first in Theorem~\ref{Sthm}. Since we can always diagonalize, this is a purely classical problem, and we thus start out with some lemmas for classical information theory. The first one is a simple lemma on randomness extraction which we will apply afterwards to approximate the typical set of one distribution by a coarse-graining of the typical set of another. We use the min-entropy, the Hartley entropy, and the R\'enyi entropy at parameter $-\infty$,
\begin{align}
	H_\infty(p) &= -\log \max_x p_x, \\
	H_0(p) &= \log |\{x\:|\: p_x > 0 \}|, \\
	H_{-\infty}(p) &= -\log \min_x p_x.
\end{align}

\begin{lemma}
	Let $(X,p)$ and $(Y,q)$ be finite probability spaces. Then there exists a map $f:X\to Y$ such that
	\beqn
	\label{l1bound}
		\norm{f_*(p) - q}_1 \leq 2^{H_0(q) - H_\infty(p)},
	\eeqn
	and
	\beqn
	\label{fibresize}
	|f^{-1}(y)| \leq 2^{H_{-\infty}(p)}\left(2^{-H_\infty(q)} + 2^{-H_{\infty}(p)}\right)
	\eeqn
	for all $y\in Y$.
	\label{coarsegrain}
\end{lemma}

Here, $f_*(p)$ is the distribution on $Y$ that one obtains by coarse-graining $p$ via application of $f$,
that is, by gathering in different sets the elements of $X$, obtaining a new distribution over a
smaller space $Y$. The $\norm{\cdot}_1$ is the total variation distance, i.e.~the classical version of the trace distance.

\begin{proof}
	We choose an arbitrary enumeration of the elements of $X$ as $x_1,\ldots,x_n$, and construct $f$ in piecemeal by defining $f(x_1),\ldots,f(x_n)$ one at a time. At the $i$-th step, we define $f(x_i)$ to be equal to an arbitrary $y\in Y$ whose probability has not yet been completely covered by the $p_x$ that lie in the preimage $f^{-1}(y)$ of the partially defined $f$, in the sense that
	\[
		q_y > \sum_{x\in f^{-1}(y)} p_x,	
	\]
	where the sum is only over those $x\in X$ for which $f(x)$ has already been defined and is equal to $y$. Finding such a $y$ is always possible since the normalisation of $p$ equals the normalisation of $q$. The crucial property of the $f$ thus constructed is that the total probability in a fibre $f^{-1}(y)$ is never significantly larger than $q_y$,
	\beqn
	\label{fibrebound}
		\sum_{x\in f^{-1}(y)} p_x \leq q_y + \max_x p_x.
	\eeqn
	This implies $|f^{-1}(y)|\cdot \min_x p_x \leq q_y + \max_x p_x$, resulting in~\eqref{fibresize}. To bound the total variation distance, we also use~\eqref{fibrebound},
	\begin{align*}
		\norm{f_*(p)- q}_1 &= \sum_y \max\left\{0, \sum_{x\in f^{-1}(y)} p_x - q_y\right\} \\
		&\leq \sum_y \max_x p_x = |Y|\cdot \max_x p_x.
	\end{align*}
	Since we can assume $q$ to have full support without loss of generality, this is the desired inequality~\eqref{l1bound}.
\end{proof}

\newcommand{\polylog}[1]{\mathrm{polylog}(#1)}

Turning to quantum information, we use the term ``size'' of a system to talk about the logarithm of its Hilbert space dimension, i.e.~the number of qubits needed to realise it, and write $S(\rho)  = - \mathrm{Tr}\left[ \rho \log \rho \right]$ for the von Neumann entropy of a state $\rho$.

\begin{theorem}[Asymptotic classification of states]
	For states $\rho$ and $\sigma$ on any quantum system of dimension $d$, the following are equivalent:
	\begin{enumerate}
		\item\label{Aequale} The states have equal entropy, $S(\rho) = S(\sigma)$.
		\item\label{Aconversion-best} There exists an ancilla system of size $O(\sqrt{n \log n})$ with state $\eta$ as well as a unitary $U$ such that
			\beqn
			\label{AepsU}
			\norm{\Tr{\mathrm{anc}}{U (\rho^{\otimes n}\otimes\eta) U^\dag} - \sigma^{\otimes n}}_1 \stackrel{n\to\infty}{\longrightarrow} 0.
			\eeqn
		\item\label{Aconversion-worst} There exists an ancilla system of size $o(n)$, with states $\eta$ and $\nu$ as well as unitaries $U$ and $V$ such that
			\beqn
			\label{AepsUV}
				\norm{\Tr{\anc}{U (\rho^{\otimes n}\otimes\eta) U^\dag - V (\sigma^{\otimes n}\otimes\nu) V^\dag}}_1 \stackrel{n\to\infty}{\longrightarrow} 0.
			\eeqn
	\end{enumerate}
	\label{Sthm}
\end{theorem}

That is, up to things that happen on an ancilla of sublinear size, two states on a system are many-copies equivalent if and only if they have the same entropy. Condition~\ref{Aconversion-best} is a set of requirements on such a many-copies equivalence that we believe to be roughly minimal; in particular, the number of qubits needed to implement the ancilla system grows only barely faster than $O(\sqrt{n})$, and in fact the particular growth rate of $O(\sqrt{n\log n})$ is an arbitrary choice and can be replaced by any function that grows faster than $\sqrt{n}$. Condition~\ref{conversion-worst}, in contrast, is a more permissive notion of many-copy equivalence that is still strong enough to imply~\ref{Aequale}, but it does not provide a different physical intuition than~\ref{Aconversion-best}.

The superoperators $\Tr{\mathrm{anc}}{U ( \cdot \otimes \eta) U^\dag}$ form channels which are close to unitary in the sense of being implementable with only a sublinear ancilla; these are precisely the channels of subexponential Kraus rank. It may be interesting to study such channels in their own right, and there may be relations to~\cite{devetak_resource_2008}.

\newcommand{\diag}[1]{\textrm{diag}(#1)}
\newcommand{\typ}{\textrm{typ}}
\newcommand{\type}{\textrm{type}}

\begin{proof}
	The implication from~\ref{Aconversion-best} to~\ref{Aconversion-worst} is trivial.

	Assuming~\ref{Aconversion-worst}, the claim $S(\rho) = S(\sigma)$ can be proven as follows. Let $D$ be the dimension of the ancilla system. Then using the fact that adding or discarding the ancilla cannot change the entropy by more than $\log(D)$\footnote{More precisely, by the fact that the conditional entropy of the ancilla given the system is at most $\log(D)$ in absolute value.}, we obtain, writing $\eps$ for the left-hand side of~\eqref{AepsUV},
	\begin{widetext}
	\begin{align*}
		|S(\rho) - S(\sigma)| & = \frac{1}{n} |S(\rho^{\otimes n}) - S(\sigma^{\otimes n})| \leq \frac{1}{n}\left| S(\rho^{\otimes n} \otimes \eta) - S(\sigma^{\otimes n} \otimes \nu)\right| + 2\frac{\log(D)}{n} \\
		& = \frac{1}{n}\left| S\left(U^\dag(\rho^{\otimes n} \otimes \eta)U\right) - S\left(V^\dag(\sigma^{\otimes n} \otimes \nu)V\right)\right| + 2 \frac{\log(D)}{n} \\
		& \leq \frac{1}{n} \left| S\left( \Tr{\anc}{U^\dag(\rho^{\otimes n}\otimes\eta)U}\right) - S\left( \Tr{\anc}{V^\dag(\sigma^{\otimes n}\otimes\nu)V} \right) \right| + 4\frac{\log(D)}{n} \\
		& \stackrel{\eqref{AepsUV}}{\leq} \frac{1}{n} \left( \log(d^n) \eps + O(1) \right) + 4\frac{\log(D)}{n} = O(\eps) + O(n^{-1}) + 4\frac{\log(D)}{n}.
	\end{align*}
	\end{widetext}
	where the last estimate is by Fannes' inequality. Since $\eps\to 0$ as $n\to\infty$ while $D$ grows only subexponentially, it follows that $|S(\rho) - S(\sigma)|$ is smaller than any positive number, and therefore $S(\rho) = S(\sigma)$.
	
	To show that~\ref{Aequale} implies~\ref{Aconversion-best}, we can assume by unitary invariance that $\rho$ and $\sigma$ are diagonal in the same basis, where they are given by $\rho = \diag{p_1,\ldots,p_d}$ and $\sigma = \diag{q_1,\ldots,q_d}$. In other words, we are in a classical situation involving finite probability spaces with distributions $p=(p_1,\ldots,p_d)$ and $q=(q_1,\ldots,q_d)$, and we therefore use classical notation and terminology for the remainder of the proof, and write $n_i$ for the number of times that outcome $i$ occurs upon sampling from $p^{\otimes n}$ or $q^{\otimes n}$. The central limit theorem guarantees that for $p^{\otimes n}$, the set of outcomes that are strongly typical in the sense that
	\beqn
	\label{rhotyp}
		 n_i \in \left[ \left( n - \sqrt{n\log n} \right) p_i , \left( n + \sqrt{n\log n} \right) p_i \right]
	\eeqn
	for every $i=1,\ldots,d$ has a total probability that approaches $1$ as $n\to\infty$. Let $T_p$ denote this typical set of outcomes and $p_\typ$ the resulting normalised distribution on typical outcomes that one obtains by conditioning on typicality. Similarly, let $T_q$ be the strongly typical set for $q^{\otimes n}$, corresponding to outcome frequencies $n_i$ restricted by
	\beqn
	\label{sigmatyp}
		n_i \in \left[ \left( n - \sqrt{n\log n} \right) q_i , \left( n + \sqrt{n\log n} \right) q_i \right],
	\eeqn
	and $q_\typ$ the associated typical distribution. By bounding the lowest and the highest probability of any outcome in this strongly typical set, it is straightforward to show the following inequalities,
	\begin{align}
	\begin{split}
		\label{typbound}
		H_0(p_\typ) &\geq H_{\infty}(p_{\typ}) \geq n S(p) \left(1  - \sqrt{\frac{\log n}{n}} \right), \\
		H_0(p_\typ) &\leq H_{-\infty}(p_\typ) \leq n S(p)\left( 1 + \sqrt{\frac{\log n}{n}} \right),
	\end{split}
	\end{align}
	where we still write $S(p) = H_1(p)$ for the Shannon entropy, and the second inequality holds for sufficiently large $n$ where the modification due to the conditioning is negligible. The analogous bounds hold for $q_\typ$. Note that the individual probabilities of the typical outcomes may still vary by a factor of up to $2^{2\sqrt{n\log n}\, S(p)}$, so that the typical distributions $p_\typ$ and $q_\typ$ may still be far from uniform. The strong typicality inequalities~\eqref{rhotyp} and~\eqref{sigmatyp} themselves will not be used again; all that we need are the R\'enyi entropy bounds~\eqref{typbound}, and that the probability of typicality approaches $1$ as $n\to\infty$. 

	Now let $r_1$ be the uniform distribution on $3 \sqrt{n\log n}\, S(p)$ many ancilla bits, rounded to the closest integer; in the following, we ignore the irrelevant rounding error. By~\eqref{typbound}, this results in the bounds
	\begin{align*}
		n S(p) + 2 \sqrt{n\log n}\, S(p) &\leq H_{\infty}(p_\typ\otimes r_1) \\
		&\leq H_{-\infty}(p_\typ\otimes r_1) \\
		&\leq n S(p) + 4 \sqrt{n\log n}\, S(p).
	\end{align*}
	Hence by Lemma~\ref{coarsegrain}, we can find a map $f:T_p\times\{0,1\}^{3\sqrt{n\log n}\, S(p)}\to T_q$ such that
	\begin{align*}
		\norm{f_*(p_\typ\otimes  r_1) - q_\typ}_1 &\leq 2^{H_0(q_\typ) - H_\infty(p_\typ\otimes  r_1)} \\
		&\leq 2^{- \sqrt{n\log n}\, S(p)}
	\end{align*}
	by~\eqref{l1bound}, which decays superpolynomially in $n$. Thanks to~\eqref{fibresize}, $f$ can be implemented using a register ancilla of dimension at most
	\begin{align*}
		&2^{H_{-\infty}(p_\typ\otimes r_1)} \left( 2^{-H_\infty(q_\typ)} + 2^{-H_\infty(p_\typ\otimes r_1)}\right) \\
		&\leq 2^{5 \sqrt{n\log n}\, S(p)} + 2^{2 \sqrt{n\log n}\, S(p)} = 2^{O(\sqrt{n\log n})},
	\end{align*}
	which is initially taken to carry an arbitrary deterministic distribution $r_2$ and gets utilised to dilate $f$ to a bijection.

	We now put $ r :=  r_1 \otimes  r_2$, so that our total ancilla still has size $O(\sqrt{n\log n})$. We take $U$ to be given by the action of $f$ on the system and first ancilla, dilated to a bijection by the second ancilla. Since $f_*$ is contractive, we have
	\begin{align*}
		\norm{f_*(p\otimes r_1) - q}_1 &\leq \norm{p_\typ - p}_1 + \norm{q_\typ - q}_1 \\
		&+ \norm{f_*(p_\typ\otimes  r_1) - q_\typ}_1 \stackrel{n\to\infty}{\longrightarrow} 0,
	\end{align*}
	since each individual term tends to $0$. This establishes~\eqref{AepsU} in classical notation.
\end{proof}

Getting to thermodynamics, we now also want to take energy preservation into account. To this end, we develop a method to turn every unitary into an energy-preserving unitary, while achieving approximately the same conversion of states. This relies on a protocol modelled after~\cite[Appendix~E]{brandao_resource_2013}. For finite sets of numbers $\mathcal{A},\mathcal{B}\subseteq\R$, we consider the sumset $\mathcal{A} + \mathcal{B} := \{\: a + b \:|\: a\in\mathcal{A},\: b\in\mathcal{B}\:\}$, as studied in additive combinatorics~\cite{breuillard_doubling_2013}, and similarly also the difference set $\mathcal{A} - \mathcal{B} = \{\: a - b \:|\: a\in\mathcal{A},\: b\in\mathcal{B}\:\}$. Furthermore, we write $\norm{\mathcal{A}} := \max_{a\in\mathcal{A}} |a|$. And from now on, we also use $\rho\approx_\eps\sigma$ as a shorthand for $\norm{\rho - \sigma}_1 \leq \eps$.

\begin{lemma}[Achieving energy preservation]
	Let $0 < \delta < 1$ and suppose that $\mathcal{L},\mathcal{M}\subseteq \R$ are finite sets of numbers such that
	\beqn
	\label{Mbound}
		|\mathcal{M} + \mathcal{L}| \leq (1+\delta) |\mathcal{M}|,\qquad
		|\mathcal{M} - \mathcal{L}| \leq (1+\delta) |\mathcal{M}|
	\eeqn
	and $\norm{\mathcal{L}}\leq\norm{\mathcal{M}}$. Given a quantum system with Hamiltonian $H$, suppose that $\rho$ and $\sigma$ are states supported on energy levels in $\mathcal{L}$, and that there is a unitary $U$ such that $U \rho U^\dag \approx_\delta \sigma$. Then there is an ancilla system of size $O(\log|\mathcal{M}|)$ with $\norm{H_\anc}\leq 4\norm{\mathcal{M}}$ and state $\eta$ as well as an energy-preserving unitary $\tilde{U}$  on the joint system, that is
\[
	\left[ \tilde{U} , H + H_\anc  \right] = 0,
\]
such that 
\[
		\Tr{\anc}{\tilde{U}(\rho\otimes\eta)\tilde{U}^\dag}\approx_{4\delta} \sigma.
\]
\label{coherence}
\end{lemma}

Interestingly, what makes this difficult to prove are the quantum coherences that $\rho$ and $\sigma$ may have between the energy levels: in the classical case in which neither $\rho$ nor $\sigma$ has any coherence across energy levels, a unitary can easily be made energy-preserving by adding an ancilla in an initial state which can absorb any energy difference that may arise.

\begin{proof}
	We do this by distinguishing two cases: first, the case that $\sigma$ has no coherences across energy levels; second, the case that $\rho$ has no such coherences. In each case, we will use an ancilla of size $O(\log |\mathcal{M}|)$ with $\norm{H_\anc}\leq 2\norm{\mathcal{M}}$ and obtain a trace distance bound of $2\delta$. This is sufficient, since in the general case we can choose an arbitrary state $\tau$ without energy coherences which has the same spectrum (with multiplicities) as that of $\rho$ or $\sigma$, and compose the protocols constructed in the two cases, first from $\rho$ to $\tau$ and then from $\tau$ to $\sigma$. This results in the claimed bounds.

	Let the spectral decomposition of the system's Hamiltonian be $H = \sum_{\lambda\in\spec(H)} \lambda P_\lambda$, with $P_\lambda$ the projection onto the corresponding energy eigenspace. 

	\begin{itemize}
	\item[\underline{Case 1:}] $\sigma$ has no coherences across energy levels, i.e.~$P_\lambda \sigma P_\mu = 0$ if $\lambda\neq\mu$.
	
		In this case, let the ancilla space be $\hil_\anc:=\C^{|\mathcal{M} - \mathcal{L}|}$ with Hamiltonian given by $H_\anc = \sum_{h\in\mathcal{M} - \mathcal{L}} h |h\rangle\langle h|$. By~\eqref{Mbound} and $\delta < 1$, the ancilla size is indeed $\log|\mathcal{M} - \mathcal{L}| = O(\log|\mathcal{M}|)$ and moreover $\norm{H_\anc}\leq 2|\mathcal{M}|$. We take the initial ancilla state to be pure $\eta := |\psi\rangle\langle\psi|$, and given by the Hadamard state
	\[
		|\psi\rangle := |\mathcal{M} - \mathcal{L}|^{-1/2} \sum_{h\in\mathcal{M} - \mathcal{L}} |h\rangle.
	\]
	Furthermore, consider the energy-preserving partial isometry
	\[
		\tilde{V} := \sum_{h\in\mathcal{M}} \sum_{\lambda,\mu\in\mathcal{L}} P_\lambda U P_\mu \otimes |h - \lambda\rangle\langle h - \mu|,
	\]
	Then $\tilde{V} (\rho\otimes\eta) \tilde{V}^\dag$ evaluates to
	\begin{widetext}
	\begin{align*}
		&|\mathcal{M} - \mathcal{L}|^{-1} \sum_{h_1,h_2\in\mathcal{M}} \sum_{\lambda_1,\lambda_2,\mu_1,\mu_2\in\mathcal{L}} \left( P_{\lambda_1} U P_{\mu_1} \otimes |h_1 - \lambda_1\rangle\langle h_1 - \mu_1| \right) \\
		& \hspace{7cm} \times \left( \rho \otimes \sum_{\ell_1,\ell_2\in\mathcal{M}-\mathcal{L}} |\ell_1\rangle\langle \ell_2|\right) \left( P_{\mu_2} U^\dag P_{\lambda_2} \otimes |h_2 - \mu_2\rangle\langle h_2 - \lambda_2| \right) \\
		= {}& |\mathcal{M} - \mathcal{L}|^{-1} \sum_{\lambda_1,\lambda_2,\mu_1,\mu_2\in\mathcal{L}} P_{\lambda_1} U P_{\mu_1} \rho P_{\mu_2} U^\dag P_{\lambda_2} \otimes \sum_{h_1,h_2\in\mathcal{M}} |h_1 - \lambda_1\rangle\langle h_2 - \lambda_2| \\
		\approx_\delta {}& |\mathcal{M} - \mathcal{L}|^{-1} \sum_{\lambda_1,\lambda_2\in\mathcal{L}} P_{\lambda_1} \sigma P_{\lambda_2} \otimes \sum_{h_1,h_2\in\mathcal{M}} |h_1 - \lambda_1\rangle\langle h_2 - \lambda_2| \\
		= {}& |\mathcal{M} - \mathcal{L}|^{-1} \sum_{\lambda\in\mathcal{L}} P_\lambda \sigma P_\lambda \otimes \sum_{h_1,h_2\in\mathcal{M}} |h_1 - \lambda\rangle\langle h_2 - \lambda|,
	\end{align*}
	\end{widetext}
	where the last step uses the assumption of absence of coherence in $\sigma$. The resulting reduced state is therefore
	\begin{align*}
		\Tr{\anc}{\tilde{V}(\rho\otimes\eta)\tilde{V}^\dag} &\approx_\delta
		\sum_{\lambda\in\mathcal{L}} \frac{|\mathcal{M}|}{|\mathcal{M} - \mathcal{L}|} P_\lambda \sigma P_\lambda \\
		&= \frac{|\mathcal{M}|}{|\mathcal{M} - \mathcal{L}|} \sigma \geq (1-\delta) \sigma.
	\end{align*}
	So if we take $\tilde{U}$ to be any energy-preserving unitary dilation of $V$, so that $\tilde{V}$ decomposes into a direct sum of $\tilde{V}$ plus some other arbitrary energy-preserving partial isometry, then the total weight of $\rho\otimes\eta$ on the orthogonal complement of the support of $\tilde{V}$ is at most $\delta$. This shows that
	\[
		\norm{ \Tr{\anc}{\tilde{U}(\rho\otimes\eta)\tilde{U}} - \sigma }_1 \leq 2\delta,
	\]
	as desired.
	\item[\underline{Case 2:}] $\rho$ has no coherences across energy levels, i.e.~$P_\lambda \rho P_\mu = 0$ if $\lambda\neq\mu$. It turns out that we can proceed very similarly.
	
	In this case, let the ancilla space be $\hil_\anc:=\C^{|\mathcal{M} + \mathcal{L}|}$ with Hamiltonian given by $H_\anc = \sum_{h\in\mathcal{M} + \mathcal{L}} h |h\rangle\langle h|$. By~\eqref{Mbound} and $\delta < 1$, the ancilla size is indeed $\log|\mathcal{M} + \mathcal{L}| = O(\log|\mathcal{M}|)$ and moreover $\norm{H_\anc}\leq 2|\mathcal{M}|$. We take the initial ancilla state to be pure $\eta := |\psi\rangle\langle\psi|$, and given by the Hadamard state
	\[
		|\psi\rangle := |\mathcal{M} + \mathcal{L}|^{-1/2} \sum_{h\in\mathcal{M} + \mathcal{L}} |h\rangle.
	\]
	Furthermore, let $\tilde{U}$ to be any energy-preserving dilation of the energy-preserving partial isometry
	\[
		\tilde{V} := \sum_{h\in\mathcal{M}} \sum_{\lambda,\mu\in\mathcal{L}} P_\lambda U P_\mu \otimes |h + \mu\rangle\langle h + \lambda|,
	\]
	so that $\tilde{U}$ decomposes into a direct sum of $\tilde{V}$ plus an arbitrary other partial isometry. Then $\tilde{V} (\rho\otimes\eta) \tilde{V}^\dag$ evaluates to
	\begin{widetext}
	\begin{align*}
		&|\mathcal{M} + \mathcal{L}|^{-1} \sum_{h_1,h_2\in\mathcal{M}} \sum_{\lambda_1,\lambda_2,\mu_1,\mu_2\in\mathcal{L}} \left( P_{\lambda_1} U P_{\mu_1} \otimes |h_1 + \mu_1\rangle\langle h_1 + \lambda_1| \right) \\
		& \hspace{7cm} \times \left( \rho \otimes \sum_{\ell_1,\ell_2\in\mathcal{M}+\mathcal{L}} |\ell_1\rangle\langle \ell_2|\right) \left( P_{\mu_2} U^\dag P_{\lambda_2} \otimes |h_2 + \lambda_2\rangle\langle h_2 + \mu_2| \right) \\
		= {}& |\mathcal{M} + \mathcal{L}|^{-1} \sum_{\lambda_1,\lambda_2,\mu_1,\mu_2\in\mathcal{L}} P_{\lambda_1} U P_{\mu_1} \rho P_{\mu_2} U^\dag P_{\lambda_2} \otimes \sum_{h_1,h_2\in\mathcal{M}} |h_1 + \mu_1\rangle\langle h_2 + \mu_2| \\
		= {}& |\mathcal{M} + \mathcal{L}|^{-1} \sum_{\mu_1,\mu_2\in\mathcal{L}} U P_{\mu_1} \rho P_{\mu_2} U^\dag \otimes \sum_{h_1,h_2\in\mathcal{M}} |h_1 + \mu_1\rangle\langle h_2 + \mu_2| \\
		= {}& |\mathcal{M} + \mathcal{L}|^{-1} \sum_{\mu\in\mathcal{L}} U P_\mu \rho P_\mu U^\dag \otimes \sum_{h\in\mathcal{M}} |h + \mu\rangle\langle h + \mu|,
	\end{align*}
	\end{widetext}
	where the last step uses the assumption of absence of coherence in $\rho$. The resulting reduced state is therefore
	\begin{align*}
		\Tr{\anc}{\tilde{V}(\rho\otimes\eta)\tilde{V}^\dag} &=
		\sum_{\mu\in\mathcal{L}} \frac{|\mathcal{M}|}{|\mathcal{M}+\mathcal{L}|} U P_\mu \rho P_\mu U^\dag \\
		&= \frac{|\mathcal{M}|}{|\mathcal{M}+\mathcal{L}|} U \rho U^\dag \\
		&\approx_\delta \frac{|\mathcal{M}|}{|\mathcal{M}+\mathcal{L}|} \sigma \geq (1-\delta)\sigma.
	\end{align*}
	The claim now follows from the same estimate as in Case 1.\qedhere
	\end{itemize}
\end{proof}

We are now sufficiently equipped to approach the proof of the main result. We write $E(\rho)  = \mathrm{Tr}\left[ H \rho \right]$ for the average energy of a state $\rho$ on a system with Hamiltonian $H$.

\begin{theorem}[Asymptotic classification of states in thermodynamics]
	For states $\rho$ and $\sigma$ on any quantum system of dimension $d$ with given Hamiltonian $H$, the following are equivalent:
	\begin{enumerate}
		\item\label{app_equalee} The states have equal entropy and average energy, $S(\rho) = S(\sigma)$ and $E(\rho) = E(\sigma)$, 
		\item\label{app_conversion-best} There exists an ancilla system of size $O(\sqrt{n\log n})$ whose Hamiltonian $H_{\anc}$ satisfies $\norm{H_{\anc}} \leq O(n^{2/3})$ with state $\eta$ as well as an energy-preserving unitary $U$ such that
			\beqn
			\label{epsU}
				\norm{\Tr{\mathrm{anc}}{U (\rho^{\otimes n}\otimes\eta) U^\dag} - \sigma^{\otimes n}}_1 \stackrel{n\to\infty}{\longrightarrow} 0.
			\eeqn
		\item\label{conversion-worst} There exists an ancilla system of size $o(n)$ whose Hamiltonian $H_{\anc}$ satisfies $\norm{H_{\anc}} \leq o(n)$ with states $\eta$ and $\nu$ as well as energy-preserving unitaries $U$ and $V$ such that
			\beqn
			\label{epsUV}
				\norm{\Tr{\anc}{U (\rho^{\otimes n}\otimes\eta) U^\dag - V (\sigma^{\otimes n}\otimes\nu) V^\dag}}_1 \stackrel{n\to\infty}{\longrightarrow} 0.
			\eeqn
	\end{enumerate}
	\label{SEthm}
\end{theorem}

\begin{definition}
When one (and hence all) of these conditions hold, we say that $\rho$ is asymptotically equivalent to $\sigma$, and we write $\rho \asmpequiv \sigma$.
\end{definition}

The bound on $\norm{H_\anc}$ in condition~\ref{app_conversion-best} is not tight: our proof adapts straightforwardly if one replaces the exponent of $2/3$ by any other exponent strictly greater than $1/2$. We expect that the bound can be reduced even more, down to at least $O(\sqrt{n\log n})$ as in Theorem~\ref{Sthm}, but proving this will probably require a more fine-grained arithmetical analysis of the energy levels.

Our interpretation of this result is essentially analogous to Theorem~\ref{Sthm}. The bound on $\norm{H_{\anc}}$ is important in that without such a bound, we could transfer an arbitrary amount of energy to or from the ancilla while only modifying the system state marginally (embezzlement). Of course, none of this is specific to the observable under consideration being energy, and the theorem applies likewise to angular momentum or to any other observable. In fact, we expect the analogous theorem to hold for any finite number of commuting observables on the system that are required to be preserved by the unitaries, with very similar proof. The case of non-commuting observables may be more difficult.

\begin{proof}
	From~\ref{conversion-worst} to~\ref{app_equalee}, equality of entropy follows from Theorem~\ref{Sthm}. Equality of energy follows from an estimate analogous to the estimate of entropy difference. With $H^{(n)}$ being the $n$-qudit Hamiltonian and writing $\eps$ for the left-hand side of~\eqref{epsUV},
	\begin{widetext}
	\begin{align*}
		|E(\rho) - E(\sigma)| & = \frac{1}{n} |E(\rho^{\otimes n}) - E(\sigma^{\otimes n})| \leq \frac{1}{n}\left| E(\rho^{\otimes n} \otimes \eta) - E(\rho^{\otimes n} \otimes \nu)\right| + 2\frac{\norm{H_{\anc}}}{n} \\
		& = \frac{1}{n}\left| E\left(U(\rho^{\otimes n} \otimes \eta)U^\dag\right) - E\left(V(\rho^{\otimes n} \otimes \nu)V^\dag\right)\right| + 2 \frac{\norm{H_{\anc}}}{n} \\
		& \leq \frac{1}{n} \left| E\left( \Tr{\anc}{U(\rho^{\otimes n}\otimes\eta)U^\dag}\right) - E\left( \Tr{\anc}{V(\sigma^{\otimes n}\otimes\nu)V^\dag} \right) \right| + 4\frac{\norm{H_{\anc}}}{n} \\
		& \leq \frac{1}{n} \eps \norm{H^{(n)}} + 4\frac{\norm{H_{\anc}}}{n}
	\end{align*}
	\end{widetext}
	Since $H^{(n)}$ is additive in $n$, we have $\norm{H^{(n)}} = n\norm{H}$, and the first term vanishes as $\eps\to 0$. The second term vanishes as $n\to\infty$ due to the assumption of sublinearity of $\norm{H_{\anc}}$. Note that the bound on $H_{\anc}$ now plays the role of the bound on entropy change due to the ancilla.
	
	To see that~\ref{app_equalee} implies~\ref{app_conversion-best}, we first apply Theorem~\ref{Sthm}. So for given $\eps>0$, we have $n\in\N$ together with the other data such that
	\beqn
	\label{noenergyprev}
		\Tr{\mathrm{anc}}{U (\rho^{\otimes n}\otimes\eta) U^\dag} \approx_\eps \sigma^{\otimes n}.
	\eeqn
	We now need to find another unitary $\tilde{U}$ that achieves something like~\eqref{noenergyprev} while also being energy-preserving.
	
	Let the spectral decomposition of the system's Hamiltonian be $H = \sum_{i=1}^\ell E_i P_i$, and let us assume that the Hamiltonian has been shifted such that $E(\rho) = E(\sigma) = 0$ for simplicity. In order to impose strong energy typicality, let us consider the state $\rho_\typ$ obtained by restricting $\rho^{\otimes n}$ such that a measurement of $P_i^{(n)}$ will result in an outcome in the range $n \tr{P_i \rho} \pm \sqrt{n\log n}$ with certainty. By taking $n$ to be large enough, we can assume $\rho^{\otimes n} \approx_\eps \rho_\typ$ by the central limit theorem. Let $\mathcal{E}_\rho$ denote the set of energy levels of $H^{(n)}$ on this typical subspace, and let us throw in their negatives and $0$ for good measure,
	\[
		\mathcal{L}_\rho := \mathcal{E}_\rho \cup (-\mathcal{E}_\rho) \cup \{0\}.
	\]
	By construction, the set $\mathcal{E}_\rho$ consists of all numbers of the form $\sum_i c_i E_i$, with integer coefficients $c_i$ that satisfy $|c_i - n\tr{P_i\rho}| \leq \sqrt{n\log n}$ for all $i$. Therefore, every number in $\mathcal{L}_\rho$ is an integer linear combination of any nonzero fixed number in $\mathcal{L}_\rho$ and the single-system energy levels $E_i$, using coefficients that are $O(\sqrt{n\log n})$. This implies that the $k$-fold Minkowski sum
	\[
		k\mathcal{L}_\rho = \underbrace{\mathcal{L}_\rho + \ldots + \mathcal{L}_\rho}_{k\textrm{ times}} 
	\]
	also contains only numbers given by some fixed number plus integer linear combinations of the energy levels $E_i$ using coefficients of size $O(k\sqrt{n\log n})$. Therefore the cardinality $|k\mathcal{L}_\rho|$ is at most polynomial, $O(\poly(nk))$.

	With $\sigma_\typ$ and $\mathcal{L}_\sigma$ defined in the analogous manner and satisfying the analogous cardinality bound, let us put $\mathcal{L} := \mathcal{L}_\rho \cup \mathcal{L}_\sigma$, which then in particular contains all the energy levels that are typical for $\rho$ or for $\sigma$. We have the bound
	\begin{align*}
		|k\mathcal{L}| &= \left| \bigcup_{j=0}^k j\mathcal{L}_\rho + (k - j)\mathcal{L}_\sigma\right| \\
		&\leq \sum_{j=0}^k |j\mathcal{L}_\rho|\cdot |(k-j)\mathcal{L}_\sigma| = O(\poly(nk)),
	\end{align*}
	so let us choose an exponent $\gamma\in\N$ and a coefficient $C>0$ such that $|k\mathcal{L}|\leq C(nk)^\gamma$ for all $k$; the particular values are not important.

	We now aim to apply Lemma~\ref{coherence} using $\mathcal{M}:=k\mathcal{L}$. To determine a suitable value of $k$, we show that if $n$ is sufficiently large, then there is $k\leq n^{1/7}$ such that
	\beqn
		\label{MA}
		|k\mathcal{L} + \mathcal{L}| \leq (1+\eps) |k\mathcal{L}|.
	\eeqn
	For if this was not the case, then we would have $|(k+1)\mathcal{L}| > (1+\eps)|k\mathcal{L}|$, which yields by induction on $k$,
	\[
		|k\mathcal{L}| \geq (1 + \eps)^k |\mathcal{L}|.
	\]
	For $k = n^{1/7}$, we would then be led to conclude
	\[
		(1 + \eps)^{n^{1/7}} |\mathcal{L}| \leq |n^{1/7}\mathcal{L}| \leq C(n^{8/7})^\gamma.
	\]
	Since the left-hand side grows superpolynomially in $n$ while the right-hand side grows only polynomially, this cannot be the case for all $n$. It follows that for suitably large $n$, there is $k\leq n^{1/7}$ such that~\eqref{MA} holds; let us fix such a $k$. We now equip the existing ancilla in~\eqref{noenergyprev} with the trivial Hamiltonian $H_\anc := 0$, so that also $\rho_\typ\otimes\eta$ is supported on the energy levels in $\mathcal{L}$. Because $\Tr{\anc}{U(\rho_\typ\otimes\eta)U^\dag}$ is $3\eps$-close to $\sigma_\typ$, which is also supported on the energy levels in $\mathcal{L}$, it follows that $U(\rho_\typ\otimes\eta)U^\dag$ itself is already $3\eps$-close to being supported on the energy levels in $\mathcal{L}$. Let us write $\hat{\rho} := \rho^{\otimes n}\otimes\eta$ and $\hat{\sigma}$ for the restriction of $U(\rho_\typ\otimes\eta)U^\dag$ to the energy levels in $\mathcal{L}$, so that $\hat{\sigma} \approx_{3\eps} U\hat{\rho}U^\dag$. By taking $\mathcal{M}:=k\mathcal{L}$ in Lemma~\ref{coherence}, we can therefore conclude the existence of an additional ancilla system $\anc'$ of size $O(\log(|k\mathcal{L}|)) = O(\log n)$ with Hamiltonian bounded by $4\norm{k\mathcal{L}} = 4 k\norm{\mathcal{L}}\leq n^{1/7}\cdot O(\sqrt{n\log n}) < O(n^{2/3})$ as claimed, with an ancilla state $\eta'$ and energy-preserving unitary $\tilde{U}$ such that
	\[
		\Tr{\anc'}{\tilde{U}(\hat{\rho}\otimes \eta')\tilde{U}^\dag} \approx_{12\eps} \hat{\sigma}.
	\]
	Putting all this together, we therefore have
	\begin{align*}
	\Tr{\anc,\anc'}{\tilde{U}(\rho^{\otimes n}\otimes\eta\otimes\eta')\tilde{U}^\dag}
	&= \Tr{\anc,\anc'}{\tilde{U}(\hat{\rho}\otimes \eta')\tilde{U}^\dag} \\
	&\approx_{12\eps} \Tr{\anc}{\hat{\sigma}} \\
	&\approx_{3\eps} \Tr{\anc}{U\hat{\rho}U^\dag} \\
	&\approx_\eps \sigma^{\otimes n},
	\end{align*}
	resulting in a total trace distance difference between the left-hand side and the right-hand side of at most $16\eps$.
\end{proof}

\end{document}